\documentclass[12pt]{article} 

\pdfoutput=1

\usepackage[margin=1in]{geometry} 

\usepackage{amsmath, amsthm, amssymb, mathtools}
\usepackage{thmtools}
\usepackage{thm-restate}
\usepackage[utf8]{inputenc}
\usepackage{enumerate}[shortlabels] 
\usepackage[
  bookmarks=true,
  bookmarksnumbered=true,
  bookmarksopen=true,
  pdfborder={0 0 0},
  breaklinks=true,
  colorlinks=true,
  linkcolor=black,
  citecolor=black,
  filecolor=black,
  urlcolor=black,
]{hyperref} 
\usepackage[table,x11names]{xcolor}
\usepackage{tikz} 
\usepackage{setspace}
\usepackage[capitalise]{cleveref}

\theoremstyle{plain}
\newtheorem{theorem}{Theorem}
\numberwithin{theorem}{section}

\newtheorem{lemma}[theorem]{Lemma}
\newtheorem{proposition}[theorem]{Proposition}

\newtheorem*{open-question}{Open Question}

\theoremstyle{definition}
\newtheorem{definition}[theorem]{Definition}
\newtheorem{example}[theorem]{Example}
\newtheorem{remark}[theorem]{Remark}
\newtheorem{fact}[theorem]{Fact}

\newcommand{\Z}{\mathbb{Z}}
\newcommand{\R}{\mathbb{R}}
\newcommand{\N}{\mathbb{N}}

\newcommand{\mc}[1]{\mathcal{#1}}
\newcommand{\mcb}[1]{\boldsymbol{\mc{#1}}}

\newcommand{\Down}{\text{Down}}

\usepackage{algpseudocode} 
\usepackage{algorithm}
\usepackage{footnote} 
\makesavenoteenv{algorithm} 

\newcommand{\problem}{\bigl(N,G,(\succ_i)_{i\in N}\bigr)}
\newcommand{\extproblem}{\bigl(N,G,(\succ_i')_{i\in N}\bigr)}
\newcommand{\anotherproblem}{\bigl(N',G',(\succ_i')_{i\in N}\bigr)}

\newcommand{\subsuccnospace}{\tilde{\succ}}
\newcommand{\subsucc}{\ \subsuccnospace\ }
\newcommand{\subproblem}{\bigl(N, G', (\subsuccnospace_i)_{i \in N}\bigr)}
\renewcommand{\epsilon}{\varepsilon}

\usepackage[backend=biber,hyperref,doi,style=authoryear,url=false,
sorting=nyt,natbib=true,giveninits=true,uniquename=init]{biblatex}
\AtEveryBibitem{
\clearfield{note}
\clearfield{isbn}
\clearfield{issn}
\clearlist{publisher}
\clearfield{series}
\clearfield{eventtitle}
\clearlist{location}
}

\title{Stable Menus of Public Goods: A Matching Problem\thanks{Fish was supported by an NSF Graduate Research Fellowship and a Kempner Institute Graduate Fellowship. Gonczarowski's research was supported by the National Science Foundation (NSF-BSF grant No.\ 2343922), Harvard FAS Dean’s Competitive Fund for Promising Scholarship, and Harvard FAS Inequality in America Initiative. The authors thank Ulrike Schmidt-Kraepelin, Alex Meiburg, Noam Nisan, Alden Rogers, Ariel Procaccia, Marco Scarsini, Ran Shorrer, and participants at the 34th Stony Brook International Conference on Game Theory for helpful comments and insightful discussions.}}

\author{Sara Fish\thanks{School of Engineering and Applied Sciences, Harvard University | \emph{E-mail}: \mbox{\href{mailto:sfish@g.harvard.edu}{sfish@g.harvard.edu}}} \and Yannai A. Gonczarowski\thanks{Department of Economics and Department of Computer Science, Harvard University | \emph{E-mail}: \href{mailto:yannai@gonch.name}{yannai@gonch.name}.} \and Sergiu Hart\thanks{Federmann Center for the Study of Rationality, Bogen Department of Economics, and Einstein Institute of Mathematics, The Hebrew University of Jerusalem | \emph{E-mail}: \href{mailto:hart@huji.ac.il}{hart@huji.ac.il}.}}

\date{June 6, 2025}

\begin{document}

\begin{titlepage}
\maketitle 

\begin{abstract}
We study a matching problem between agents and public goods, in settings without monetary transfers. Since goods are public, they have no capacity constraints. There is no exogenously defined budget of goods to be provided. Rather, each provided good must justify its cost by being utilized by sufficiently many agents, leading to strong complementarities in the ``preferences'' of goods. Furthermore, goods that are in high demand given other already-provided goods must also be provided. The question of the existence of a stable solution (a menu of public goods to be provided) exhibits a rich combinatorial structure. We uncover sufficient conditions and necessary conditions for guaranteeing the existence of a stable solution, and derive both positive and negative results for strategyproof stable matching.
\end{abstract}

\thispagestyle{empty}
\end{titlepage}

\setcounter{tocdepth}{2}
\tableofcontents
\thispagestyle{empty}
\clearpage

\setcounter{page}{1}

\section{Introduction}\label{sec:introduction}

\textbf{Which courses should an online university offer?} Courses have no cap on participation, but were the university to offer a non-demanded course, it would not recoup the cost of the course. On the other hand, were it to fail to offer a very demanded course, students would object. \textbf{Where throughout a city would we expect to find ATMs (or vending machines)?} There is virtually no limit on the amount of users of a single machine, however if a machine is in a spot that is not sufficiently convenient for many people, this machine would not cover its cost. On the other hand, if a spot that is convenient to many people does not have a machine, some company would respond to the opportunity. Finally, \textbf{how should a group of hikers choose which trails to hike in small groups?} On the one hand, heading out in too small of a group might not be safe. On the other hand, if a critical mass of people would rather hike another trail not included in the current plan, they would change plans and head together to that trail instead. Each of these problems can be viewed as a matching problem, of agents---students, customers, or hikers---to public goods---courses, ATMs/vending machines, or trails.

In each of the above matching problems, there are \emph{no (upper) constraints on capacity} for (i.e., no scarcity of ``usage slots'' for) public goods. However, there is a firm minimum usage on each provided public good; that is, the ``preferences'' of public goods over individuals exhibit strong \emph{complementarities}. Each of these two features---no capacity constraints on the one hand, and complementarities on the other hand---makes this problem starkly different than much of the literature on stable matching or stable assignments. And yet, it is natural to ask in this scenario how one might formalize the stability of a given offering of public goods, whether such a stable offering of public goods is guaranteed to exist, and, if and when it is guaranteed, whether there are any strategyproof stable mechanisms. In this paper, we initiate the study of these questions.

\paragraph{Defining Stability.} We consider a setting with a finite number of public goods $G = \{1, \dots, g \}$ that might be provided, and a finite number of unit-demand agents, each with a strict preference list over the public goods. A menu of public goods is a set of public goods that are provided. If a specific menu is provided, then each agent consumes from their favorite public good from this menu. Such a menu is $t$-feasible if when provided, each public good from the menu is consumed by at least $t$ agents.\footnote{In the courses example, this translates to the inability to charge more for some courses over others. In the ATM example, we thus capture settings in which there are fixed fees for ATM usage. In the vending machine example, we thus capture settings in which the prices of products sold are similar across machines.} This is a desideratum that pushes for less public goods to be provided, and becomes more restrictive as $t$ grows. A menu is $u$-contestable if there exists a public good not on the menu that $u$ agents prefer to all public goods on the menu. A menu is $u$-\emph{{un}contestable} if it is not $u$-contestable.\footnote{In the courses example, one might imagine a complaint by $u$ students as something that requires the university's attention. In the ATM example, $u$-contestability might result in market forces adding a machine.} This is a desideratum that pushes for more public goods to be provided, and becomes more restrictive as $u$ shrinks. A menu is $(t,u)$-stable if it is both $t$-feasible and $u$-uncontestable. The notion of $(t,u)$-stability captures the idea that a public good covers its cost when it is used by at least $t$ agents, and that if $u$ or more agents demand an unprovided good, then market forces will lead to its offering.\footnote{Stability takes no stand on whether the public goods that are provided (or this hypothetical addition public good) are provided by the same provider or by different providers.}

\paragraph{Existence.} We start with the case of no setup costs or other frictions, i.e., with the case of $(t,t)$-stability.\footnote{As for the case of $t>u$, beyond being less economically motivated, we note that it is immediate to see that in this case, a $(t,u)$-stable solution is not guaranteed to exist; indeed, consider one public good and $u$ agents who desire it.} While for $g=2$ a $(t,t)$-stable solution always exists, we show by an example that with $g\ge3$ public goods, there are menu selection problems and corresponding $t$ such that no $(t,t)$-stable solution exists. Our first main question is for every $g$ and $t$, therefore, how small (compared to~$t$) can $u$ be so that every menu selection problem with $g$ goods has a $(t,u)$-stable solution. As it turns out, the answer is intimately connected to the ratio between $u\!-\!1$ and~$t\!-\!1$, as already hinted by the following two relatively straightforward bounds:\footnote{While this might suggest a reparametrization of the problem using $t':=t\!-\!1$ and $u':=u\!-\!1$, we avoid such a reparametrization as we find that the quantity $t'$---the maximum number of agents that is \emph{in}sufficient for sustaining the offering of a single public good---has a lesser direct economic meaning than $t$.}

\newcounter{lb-footnote}
\setcounter{lb-footnote}{0}

\begin{restatable}{proposition}{lowerboundtwo}\label{lower-bound-2}
For every $g\ge 3$ and $t,u \in \N$ such that $\frac{u-1}{t-1} < 2$,\footnote{\ifnum\value{lb-footnote}=0 \addtocounter{lb-footnote}{1}Note \else Recall \fi that $\frac{u-1}{t-1} < 2$ is equivalent to $u \le 2t-2$.} there exists a menu selection problem with $g$ public goods that has no $(t,u)$-stable menu.
\end{restatable}

\begin{restatable}{proposition}{upperboundgminusone}\label{upper-bound-g-minus-1}
For every $g\ge 2$ and $t,u \in \N$ such that $\frac{u-1}{t-1} \ge g\!-\!1$, every menu selection problem with $g$ public goods has a $(t,u)$-stable menu.
\end{restatable}

Trying to improve either of these bounds reveals intricate mathematical structure that must be mapped out and understood. Our first main result is that the lower bound in \cref{lower-bound-2} is tight for all $g\le6$.

\begin{theorem}[See \cref{thm:positive-four,thm:positive-six}]\label{upper-bound-2}
For every $3\le g\le 6$ and $t,u \in \N$ such that $\frac{u-1}{t-1} \geq 2$, every menu selection problem with $g$ public goods has a $(t,u)$-stable menu.
\end{theorem}

We prove \cref{upper-bound-2} analytically for $g=3,4$. For $g=5,6$, we first use structural insights to reduce the question of an existence of a counterexample to a polyhedra covering problem that we encode as an SMT problem, and then computationally show that no counterexample exists via an SMT solver.

Surprisingly, running the same SMT problem for $g=7$ reveals a counterexample showing that \cref{lower-bound-2} is no longer tight when $g\ge7$.\footnote{Ascertaining the tightness of the lower bound in \cref{lower-bound-2311}, or alternatively further tightening it, whether for $g=7$ or for higher values of $g$, seems beyond the computational capabilities at our disposal. Solving the same problem encoded as an ILP turns out to be even more computationally demanding (based on performance comparisons on lower values of $g$).}

\begin{restatable}{theorem}{lowerboundcomplex}\label{lower-bound-2311}
For every $g\ge 7$ and $t,u \in \N$ such that $u \leq 23 \lfloor \frac{t-1}{11} \rfloor $, there exists a menu selection problem with $g$ public goods that has no $(t,u)$-stable menu. (When $t\!-\!1$ is divisible by 11, this condition is equivalent to $\frac{u-1}{t-1} < \frac{23}{11} \approx 2.1$.)
\end{restatable}

Our second main result is a somewhat improved upper bound for all $g \ge 7$.
\begin{theorem}[For more details, see \cref{upper-bound-g-minus-2-formal}]\label{upper-bound-g-minus-2}
For every $g\ge 7$ and $t,u \in \N$ such that $\frac{u-1}{t-1} \geq g\!-\!2$, every menu selection problem with $g$ goods has a $(t,u)$-stable menu.
\end{theorem}

\paragraph{Strategyproofness.} In the second part of this paper, we focus on $g\in\{2,3,4,5,6\}$, i.e., the case in which we have a tight characterization for guaranteed existence, and ask for each such $g$ whether there exists a strategyproof stable mechanism, i.e., a mechanism whose inputs are a menu selection problem with $g$ public goods and with $\frac{u-1}{t-1} \ge 2$,\footnote{Or $\frac{u-1}{t-1} \ge 1$ when $g=2$.} and whose output is a $(t,u)$-stable menu, which is strategyproof when its output is viewed as the matching that matches each agent to their favorite public good in the output menu. Due to the nature of menu selection problems and the motivation behind them, another desideratum that we look for is for the mechanism to be anonymous, i.e, not factor in the identities of the different agents when selecting a stable menu. We first find a positive result.
\begin{theorem}[For more details, see \cref{sp-2}]\label{sp-2-informal}
For $g=2$, there exists a strategyproof anonymous stable mechanism.
\end{theorem}

The mechanism constructed in \cref{sp-2-informal} is based upon majority voting, carefully tweaked with respect to when both public goods should be provided and when only a subset of them should be provided, on the one hand maintaining stability, and on the other controlling transitions between providing different numbers of public goods to remove any incentives to misreport. Our main result for strategyproofness is a negative result for higher numbers of goods.

\begin{theorem}[For more details, see \cref{sp-3456}]\label{sp-3456-informal}
For $g \in \{3,4,5,6\}$, every anonymous stable mechanism is not strategyproof.
\end{theorem}

The major workhorse behind \cref{sp-3456-informal} is the Gibbard--Satterthwaite theorem \citep{gibbard_manipulation_1973,satterthwaite_strategy-proofness_1975}, and the main challenge is how to transform a voting problem with $g$ alternatives into a menu selection problem such that \emph{all stable menus are singletons}. As it turns out, it is possible to construct such a transformation, where furthermore stability of the solution implies unanimity.

\vspace{1em}

This paper defines the stable menu selection problem, showcases that its solutions possess an intricate structure, providing nontrivial upper and lower bounds (tight for some $g$) for guaranteed existence, and maps out its strategic properties. While \cref{upper-bound-2} completely resolves the guaranteed existence problem for $g\le 6$, \cref{lower-bound-2311,upper-bound-g-minus-2} expose the complex structure of this problem for higher number of goods, and leave an asymptotic gap of constant vs.\ $\Theta(g)$ for this question, which we leave as an open problem.

\begin{open-question}
For every $g\ge7$ and $t \in \N$, what is the minimal value $u_g(t)$ such that every menu selection problem with $g$ public goods has a $\bigl(t,u_g(t)\bigr)$-stable solution? In particular, what is the asymptotic dependence of $\max_{t}\frac{u_g(t)-1}{t-1}$ on $g$?
\end{open-question}

\subsection{Related Work}
The stable matching literature began with the seminal paper of \citet{gale_college_1962}, showing the existence of a stable matching for every 1-to-1 stable matching problem. It has been shown for various many-to-1 flavors of this problem that existence is maintained so long as the preferences on each side considers the entities on other side as substitutes \citep[see, e.g.,][]{hatfield_matching_2005,hatfield_substitutes_2010,hatfield_hidden_2015}.
Strategic considerations for stable matching were first investigated by \citet{dubins_machiavelli_1981} and \citet{roth_economics_1982}, who showed that the stable matching mechanism introduced by \citeauthor{gale_college_1962} is strategyproof for one of the sides of the market. As noted, two major differences from the problems studies in this literature are the lack of (upper) capacity constraints on the public goods, and the strong complementarities exhibited by the ``preferences'' of the goods over sets of agents (only ``preferring'' a set of agents to the empty set if it is of sufficiently large cardinality, which in turn translates to a nonmonotone rejection function).

The motivation for our problem is related to that of the literature on public projects \citep[see, e.g.,][]{papadimitriou_hardness_2008}, with a main difference being that no transfers are allowed in our model. The motivation is also related to that of facility location games initiated by \citet{hotelling_stability_1929}, for which various solution concepts have been studied \citep[for a survey see, e.g.,][]{ijcai2021p596}.
There are also connections with the literature on committee selection \citep{aziz_justified_2017,jiang_approximately_2020}, with a main difference being that in our model there is no exogenously determined ``budget'' of how many public goods could be offered, but rather public goods can be offered as long as their usage justifies it.

Finally, our notion of $u$-uncontestability is closely related to notions of core stability from cooperative game theory \citep[see, e.g.,][]{peleg_introduction_2007}, especially within the context of non-transferrable utility (NTU) games. See \cref{app:core-ntu} for a detailed treatment of our connections to this area, and in particular for a proof that \emph{balancedness}---the canonical workhorse for showing nonemptiness of cores of NTU games---falls short of characterizing existence of $(t,u)$-stable menus in menu selection problems. Strategyproof selection from the core of generalized indivisible goods allocation problems has been studied by \citet{sonmez_strategy-proofness_1999}. In contrast to the setting of that paper, which requires the core correspondence to have a property that they call being essentially single-valued, the core correspondence of our menu selection problems is not essentially single-valued (see \cref{sec:strategic} for details).

\section{Model and Notation}\label{sec:model-and-notation}

In this section, we formally introduce the stable menu selection problem. As will be seen below, we allow for agents to have incomplete preference lists over public goods. While strictly speaking this makes the problem harder, many of our results and proofs hold verbatim for this setting as well, and we do not know how to significantly shorten these proofs for the case of only complete preference lists (see also \cref{lemma:suff-complete} in \cref{app:lemmas}). For the results from the introduction for which our proofs rely on complete preference lists, this is explicitly stated as part of the detailed restatement of these results.\footnote{Specifically, we show \cref{lower-bound-2,upper-bound-g-minus-1,upper-bound-2,lower-bound-2311} even in the case of incomplete preference lists. However, our proofs for \cref{upper-bound-g-minus-2,,sp-2-informal,sp-3456-informal} only work for the case of complete preference lists (see their respective detailed restatements---\cref{upper-bound-g-minus-2-formal,sp-2,sp-3456}).}

\begin{definition}[Menu Selection Problem]\label{def:cap}
A \emph{menu selection problem} (with unit demand) is a triplet $\problem$, where:
\begin{itemize}
\item $N=\{1,\ldots,n\}$ is a finite set of \emph{agents}. A generic element of $N$ will be denoted $i$.
\item $G=\{1,\ldots,g\}$ is a finite set of \emph{public goods}. A generic element of $G$ will be denoted $j$.
\item $\succ_i$ is an order over $G$ of the form $j_1 \succ \dots \succ j_{g'}$, where $g' \leq g$, and for $j_k \in G$ distinct. (In other words, agent $i$ ranks their top $g'$ choices, with $g'$ potentially depending on $i$, and any unranked choice is as good as some outside option.) 
\end{itemize}
\end{definition}

Given a menu of public goods that are offered, each agent is assigned to their favorite of these public goods, if such a public good exists. Formally, we define the agent assignment as follows.

\begin{definition}[Agent Assignment]
Given a menu selection problem $\problem$, for a menu $O \subseteq G$ (of goods to be \textit{offered}), we define the corresponding \textit{agent assignment} $a_O : N \to O \cup \{ \perp \}$ (here $\perp$ is the ``outside option'') as follows. For each agent $i \in N$, let $G'$ denote the public goods ranked in ${\succ}_i$. If $O \cap G' \neq \emptyset$, let $a_O(i)$ be agent $i$'s favorite public good (according to ${\succ}_i$) in $O$. Otherwise, if $O \cap G' = \emptyset$, then let $a_O(i) := \perp$.
\end{definition}

We now define the two main desiderata for a menu of public goods. The first, feasibility, pushes toward offering fewer public goods, while the second, uncontestability, pushes toward offering more public goods. This tension is at the heart of the problem that we describe.

\begin{definition}[Feasibility; Uncontestability]
Let $\problem$ be a menu selection problem, and let $O \subseteq G$ be a menu of public goods to be offered.
\begin{itemize}
\item
For some $t\in\mathbb{N}$, the menu $O$ is said to be \emph{$t$-feasible} if $\bigl|a_O^{-1}(j)\bigr| \ge t$ for every $j \in O$. That is, every offered public good $j \in O$ is assigned to at least $t$ agents.
\item
For some $u\in\mathbb{N}$, the menu $O$ is said to be \emph{$u$-uncontestable} if there exists no unoffered public good $j\in G\setminus O$ such that $\bigl|a^{-1}_{O \cup \{ j \}}(j)\bigr| \ge u$.
That is, no unoffered public good has a lobby of at least $u$ agents who prefer it over all public goods in $O$.
\end{itemize}
\end{definition}

\begin{definition}[Stability]
Let $\problem$ be a menu selection problem. A menu of public goods $O \subseteq G$ is said to be \textit{$(t,u)$-stable} if it is $t$-feasible and $u$-uncontestable.
\end{definition}

Note that it is possible for an empty menu $O = \emptyset$ to be stable. Indeed, if $n < t$, then the empty menu is the unique stable menu. 

\begin{fact}[Monotonicity]\label{lemma:monotonicity}
Let $\problem$ be a menu selection problem, and let $O \subseteq G$ be a menu of public goods to be offered.
\begin{itemize}
\item
If $O$ is $t$-feasible for some $t\in\mathbb{N}$, then $O'$ is $t'$-feasible for all $O'\subseteq O$ and $t'\le t$.
\item
If $O$ is $u$-uncontestable for some $u\in\mathbb{N}$, then $O'$ is $u'$-uncontestable for all $O'\supseteq O$ and $u'\ge u$.
\end{itemize}
\end{fact}

Finally we fix some notation. Let $|j| := |a_G^{-1}(j)|$ denote the number of agents whose favorite public good is $j$. Let $|j \succ j'| := |a_{\{j,j'\}}^{-1}(j)|$ denote the number of agents who prefer public good $j$ over public good $j'$. And more generally, for $j \in G \setminus O$, let $|j \succ O| := |a_{O \cup \{ j \} }^{-1}(j)|$ denote the number of agents who prefer public good $j$ over all public goods in $O \subseteq G$.

\section{Simple Lower and Upper Bounds}\label{sec:simple-lower-upper-bounds}

Our main question is for which menu selection problems a stable solution exists. As it turns out, an interesting (and useful) point of view into this question is to ask for which combination of parameters $t,u \in \N$, and more specifically, for which ratio between $u\!-\!1$ and $t\!-\!1$, a stable solution is guaranteed to exist for all assignment problems. We start by demonstrating that for some such combinations of parameters, the existence of a stable solution is not guaranteed. Here we split into two cases: first, $g=2$ and $\frac{u-1}{t-1} < 1$, and second, $g \ge 3$ and $\frac{u-1}{t-1} < 2$.\footnote{We write the conditions on the parameters $t,u$ in terms of the ratio $\frac{u-1}{t-1}$ for ease of legibility. In doing so, we implicitly assume $t \ge 2$ and $u \ge 1$. These inequalities can also be rewritten as linear (Diophantine) inequalities: first, $\frac{u-1}{t-1} < 1$ if and only if $u \le t-1$, and second, $\frac{u-1}{t-1} < 2$ if and only if $u \le 2t-2$. These linear inequalities have the advantage of being well-defined for slightly expanded parameter regimes (all $t,u \in \mathbb{Z}^{\ge 0}$). Indeed, our results and proofs continue to (trivially) hold in this expanded parameter regime.}

\begin{proposition}\label{prop:lower-bound-g-2}
For every $t,u \in \N$ such that $\frac{u-1}{t-1} < 1$,\footnote{Recall that $\frac{u-1}{t-1} <1$ is equivalent to $u < t$.} there exists a menu selection problem on two public goods that has no $(t,u)$-stable menu.
\end{proposition}
\begin{proof}
Consider a menu selection problem with $u$ agents, all of whom have the preference list $1 \succ 2$. Then the menu $\emptyset$ is not $u$-uncontestable, because all $u$ agents would lobby for public good 1 (or 2). However, any nonempty menu $O \subseteq \{1,2\}$ is not $t$-feasible, because there are only $u < t$ agents.
\end{proof}

\lowerboundtwo*

\begin{proof}

Let $x:=\bigl\lceil\frac{u}{2}\bigr\rceil<t$. Set $n:= 3x$, and set the preferences of the agents as follows: 
\begin{itemize}
\item
$1\succ_i 2 \succ_i 3$ for every $i\in\{1,\ldots,x\}$,
\item
$2\succ_i 3 \succ_i 1$ for every $i\in\{x\!+\!1,\ldots,2x\}$,
\item
$3\succ_i 1 \succ_i 2$ for every $i\in\{2x\!+\!1,\ldots,3x\}$.
\end{itemize}

First, observe that no $O \not\subseteq \{ 1,2,3 \}$ can be $t$-feasible, because no agent will use public good $k \in O \setminus \{ 1,2,3 \}$. So, a $(t,u)$-stable menu $O$ must satisfy $ O\subseteq \{1,2,3\}$. Reasoning by cases, we now show that no such $O$ is both $t$-feasible and $u$-uncontestable.
\begin{itemize}
\item
$|O|=3$: as the number of agents $n$ satisfies $n=3x<3t$, we have that $O$ is not $t$-feasible.
\item
$|O|=2$: by symmetry, assume without loss of generality that $O=\{1,2\}$. By definition of $a_O$, we have $a_O^{-1}(2)=\{x\!+\!1,\ldots,2x\}$, and so $\bigl|a_O^{-1}(2)\bigr|=x<t$. Thus, $O$ is not $t$-feasible.
\item
$|O|=1$: by symmetry, assume without loss of generality that $O=\{1\}$. We have $3 \succ_i 1 = a_O(i)$ for every $i\in U:=\{x\!+\!1,\ldots,3x\}$, while $|U|=2x\ge u$, and so $O$ is not $u$-uncontestable. 
\item $|O|=0$: as $n = 3x \geq u$, we have that $O = \emptyset$ is not $u$-uncontestable.\qedhere
\end{itemize}
\end{proof}

We conclude this section with a fairly simple upper bound, showing that if the ratio between $u-1$ and $t-1$ is sufficiently large, then a stable solution is guaranteed to exist.

\upperboundgminusone*

\begin{proof}
Let $O := \{ k \in G:\ |k| \geq t \}$. There are two cases.

Case 1: $O \neq \emptyset$. By construction, $O$ is $t$-feasible. Additionally, each public good $j \in G \setminus O$ satisfies $|j| < t$, so the maximum possible lobby size is 
\[ \max_{j \in G \setminus O} |j \succ O| \leq (g-|O|)(t\!-\!1) \leq (g\!-\!1)(t\!-\!1) < u.\] 
Thus $O$ is $(t,u)$-stable. 

Case 2: $O = \emptyset$. If $\emptyset$ is $u$-uncontestable, then $\emptyset$ is $(t,u)$-stable, because $\emptyset$ is trivially $t$-feasible. Otherwise, there exists some $k \in G$ such that $|k \succ \emptyset| \geq u$, so $\{ k \}$ is $t$-feasible. And if $\{ k \}$ is offered, the maximum possible lobby size is 
\[ \max_{j \in G \setminus \{ k\} } |j \succ \{ k \}| \leq (g\!-\!1)(t\!-\!1) < u,\]
so $\{ k \}$ is $(t,u)$-stable.
\end{proof}

Note that the gap between the lower bound of \cref{lower-bound-2} and the upper bound of \cref{upper-bound-g-minus-1} is stark: the ratio between $u$ and $t$ is $\Theta(1)$ in the former, while it is $\Theta(g)$ in the latter. In the remainder of this paper, we explore the space between these two simple bounds.

\section{Tightness of the Lower Bound for \texorpdfstring{$g\le 6$}{g<=6}}\label{sec:small-c}

In this section, possibly surprisingly, we show that the simple lower bound of \cref{lower-bound-2} is in fact precisely tight for all $g\le 6$. That is, the condition presented in \cref{lower-bound-2} is not only sufficient but also necessary for the inexistence of a solution for such $g$. We start by proving this tightness analytically for $g\le 4$ (\cref{sec:analytic-c4-proof}), and then identify an equivalence to a polyhedra covering problem, which allows us to prove this tightness computationally all the way up to $g=6$ (\cref{sec:computational-c56-proof}). Interestingly, this tightness ceases to hold starting at $g=7$; this will be shown in \cref{sec:beyondc6}, where we further study the domain of $g\ge7$.

\subsection{Analytic Proof for $g\le 4$}\label{sec:analytic-c4-proof}

In this section we prove the first two parts of \cref{upper-bound-2} from the introduction. That is, we prove the tightness of the lower bound of \cref{lower-bound-2} in the case $g \leq 4$. Specifically, we prove the following. 

\begin{theorem}\label{thm:positive-four}
For every $g \in \{3,4 \}$ and $t,u \in \N$ such that $u \ge 2t\!-\!1$, every menu selection problem with $g$ public goods has a $(t,u)$-stable menu.
\end{theorem}

Our proof idea is as follows. One simple way to find a $(t,u)$-stable menu $O$ is to run a greedy algorithm: repeatedly add public goods demanded by $\geq u$ agents, and remove public goods only used by $<t $ agents, until $t$-feasibility and $u$-uncontestability are both satisfied. If this algorithm stops, then the resulting menu is $(t,u)$-stable. However, this algorithm may not stop: In this case, by finiteness, it will cycle. As we will see, even for $g=4$ this algorithm may indeed cycle (see \cref{example:c4-cycle}), however an analysis of the cycle will allow us to nonetheless find a $(t,u)$-stable solution for $g=4$.\footnote{This approach turns out to be insufficient for the case of $g \ge 5$. See \cref{app:greedy} for a demonstration.}

We begin by formally specifying the greedy algorithm (\cref{alg:greedy}). 

\begin{algorithm}
\caption{Greedy algorithm for the menu selection problem}\label{alg:greedy}
\begin{algorithmic}
\State \textbf{Inputs: }menu selection problem $\problem$, parameters $t,u \in \N$.
\State Initialize $O\gets \emptyset$.\footnote{This can be replaced with any other menu of public goods with no change to the analysis below.} 
\While{$O$ is not $(t,u)$-stable}
\If{$O$ is not $t$-feasible}
\State By definition there exists $j\in O$ such that $|j \succ O \setminus \{ j \} | <t$.
\State Let $j$ be a minimal\newcounter{minimalfootnote}\setcounter{minimalfootnote}{\thefootnote}\footnote{This can be replaced with any other consistent tie-breaking with no change to the analysis below.} such $j$, and update $O\gets O\setminus\{j\}$.
\ElsIf{$O$ is not $u$-uncontestable\footnote{This item can be consistently swapped with the preceding one with no change to the analysis below.}}
\State By definition there exists $j\in G\setminus O$ such that $|j\succ O|\ge u$. 
\State Let $j$ be a minimal\newcounter{savedcurrentfootnote}\setcounter{savedcurrentfootnote}{\thefootnote}\setcounter{footnote}{\theminimalfootnote}\footnotemark\setcounter{footnote}{\thesavedcurrentfootnote} such $j$, and update $O\gets O\cup\{j\}$.
\EndIf
\EndWhile\\
\Return $O$ (which is, by construction, a $(t,u)$-stable menu). 
\end{algorithmic}
\end{algorithm}
We denote a cycle of \cref{alg:greedy} by, e.g., 
\[ O_1 \xrightarrow{-j_1} O_2 \xrightarrow{+j_2} O_3 \xrightarrow{+j_3} \cdots O_\ell \xrightarrow{+j_\ell} O_1,\]
where $O_2=O_1\setminus\{j_1\}$, $O_3=O_2\cup\{j_2\}$, etc., and refer to this expression as a \textit{cycle transcript}. 

Next we give an example for when \cref{alg:greedy} may cycle in the case $g=4$ (which straightforwardly generalizes to larger $g$). 

\begin{example}[Greedy algorithm may cycle when $g=4$]\label{example:c4-cycle}
Let $\problem$ be a menu selection problem with $g=4$ and the following agents:
\begin{itemize}
    \item $t\!-\!1$ agents have preference $4 \succ 3 \succ 2 \succ 1$
    \item $t\!-\!1$ agents have preference $3 \succ 2 \succ 1 \succ 4$
    \item $t\!-\!1$ agents have preference $2 \succ 1 \succ 4 \succ 3$ 
    \item $t\!-\!1$ agents have preference $1 \succ 4 \succ 3 \succ 2$ 
\end{itemize}
Set $u := 2t\!-\!1$. Then the $(t,u)$-stable menus are $\{1,3 \}$ and $\{2,4 \}$. However, \cref{alg:greedy} cycles in the following way:
\[ \emptyset \xrightarrow{+1} \{ 1 \} \xrightarrow{+2} \{1,2\} \xrightarrow{-1} \{ 2 \} \xrightarrow{+3} \{2,3 \} \xrightarrow{-2} \{3 \} \xrightarrow{+4} \{3,4\} \xrightarrow{-3} \{4 \} \xrightarrow{+1} \{4,1\} \xrightarrow{-4} \{1\} \to \dots \]
\end{example}

The cycle in \cref{example:c4-cycle} has a simple structure: it alternates between menus of size 1 and 2. The following two 
lemmas show that (in particular) when $g = 4$, all cycles must be of this form. 

\begin{lemma}
Let $t,u \in \N$ such that $u \ge t$, and let $\problem$ be a menu selection problem with $g \ge 2$. Suppose \cref{alg:greedy} results in a cycle $O_1 \xrightarrow{-j_1} O_2 \xrightarrow{+j_2} O_3 \xrightarrow{+j_3} \cdots O_\ell \xrightarrow{+j_\ell} O_1$. Then $|O_j| \ge 1$ for all $j \in [\ell]$.
\end{lemma}
\begin{proof}
Suppose for the sake of contradiction that $\emptyset$ appears in the cycle. Then the cycle transcript contains a step $\{j \} \xrightarrow{-j} \emptyset$ for some $j \in G$. Because $j$ was removed along the cycle, it must also be added: thus, there exists some $A \subseteq G$ with $j \notin A$ such that $A \xrightarrow{+j} A \cup \{ j \}$ appears in the cycle transcript. Now, the $\{ j \} \xrightarrow{-j} \emptyset$ step implies $|j \succ \emptyset| < t$, and the $A \xrightarrow{+j} A \cup \{ j \}$ step implies $|j \succ A| \geq u$. But $|j \succ \emptyset| \geq |j \succ A|$, and because $u \geq t$, we have a contradiction. 
\end{proof}

\begin{lemma}\label{no-large-menus}
Let $t,u \in \N$ such that $u \ge 2t\!-\!1$, and let $\problem$ be a menu selection problem with $g \ge 4$. Suppose \cref{alg:greedy} results in a cycle $O_1 \xrightarrow{-j_1} O_2 \xrightarrow{+j_2} O_3 \xrightarrow{+j_3} \cdots O_\ell \xrightarrow{+j_\ell} O_1$. Then $|O_j| \le g\!-\!2$ for all $j \in [\ell]$.
\end{lemma}

\begin{proof}
Assume for contradiction that a menu with cardinality at least $g\!-\!1$ exists along a cycle. So, there exists a menu $A$ with  $|A|\ge g\!-\!2$ and a public good $j\notin A$ such that $j$ is added to $A$ along that cycle. Since $j$ is added along that cycle, it is also removed from some menu $B$ along that cycle. Thus, $|j \succ A| \geq u$ and $|j \succ B| < t$. 

We cannot have $B \subseteq A$, because this would imply $t > |j \succ B| \geq |j \succ A| \geq u$, a contradiction. Thus $B \subsetneq A$. Since $|A| = g\!-\!2$ and $j \notin B$, we have $|B \setminus A| = 1$. Let $k \in B \setminus A$ denote this unique element. 

Since $|j \succ A| \geq u$, $|j \succ B| < t$, and $u \geq 2t\!-\!1$, there exist at least $u-(t\!-\!1) \geq t$ agents which satisfy $j \succ A$ but not $j \succ B$. The only way for an agent satisfying $j \succ A$ to violate $j \succ B$ is for them to violate $j \succ k$. Thus $|k \succ j \succ A| \geq t$, which implies $|k \succ G \setminus \{ k \}| \geq t$. In other words, at least $t$ agents prefer $k$ the most; therefore, $k$ is always assigned to at least $t$ agents if offered. But $k \notin A$ while $k \in B$, so $k$ is removed along that cycle---a contradiction. 
\end{proof}

\cref{no-large-menus} implies that when $g = 4$, if \cref{alg:greedy} cycles, the cycle must be of the form
\[ \{ j_1 \} \to \{ j_1, j_2 \} \to \{ j_2 \} \to \{ j_2, j_3 \} \to \dots \to \{ j_{\ell-1} \} \to \{ j_{\ell - 1}, j_\ell \} \to \{ j_{\ell} \} \to \{ j_\ell, j_1 \} \to \{j_1 \},\]
where $\ell \in \{2,3,4\}$. It remains to analyze these cases. In \cref{lemma:12snake-len2} and \cref{lemma:12snake-len3}, we show that $\ell=2,3$ are impossible. Finally, in \cref{lemma:12snake-len4}, we show that when $\ell=4$, a $(t,u)$-stable menu can be computed from the cycle transcript directly. 

\begin{lemma}\label{lemma:12snake-len2}
Let $t,u \in \N$ such that $u \ge t$, and let $\problem$ be a menu selection problem with $g \ge 2$. Then \cref{alg:greedy} cannot cycle in the following way, for distinct $j_1, j_2 \in G$:
\[ \{ j_1 \} \to \{ j_1, j_2 \} \to \{ j_1 \}.\]
\end{lemma}
\begin{proof}
Suppose such a cycle exists. The step $\{ j_1 \} \to \{ j_1, j_2 \}$ implies $|j_2 \succ j_1| \geq u$, and the step $\{ j_1, j_2 \} \to \{ j_1 \}$ implies $|j_2 \succ j_1| < t$. This cannot happen if $u \geq t$. 
\end{proof}

\begin{lemma}\label{lemma:12snake-len3}
Let $t,u \in \N$ such that $u \ge 2t\!-\!1$, and let $\problem$ be a menu selection problem with $g \ge 3$. Then \cref{alg:greedy} cannot cycle in the following way, for distinct $j_1, j_2 , j_3\in G$:
\[ \{ j_1 \} \to \{ j_1, j_2 \} \to \{ j_2 \} \to \{ j_2, j_3 \} \to \{ j_3 \} \to \{ j_3, j_1 \} \to\{ j_1 \}.\]
\end{lemma}
\begin{proof}
Suppose such a cycle exists. The step $\{ j_1 \} \to \{ j_1, j_2 \}$ implies $|j_2 \succ j_1| \geq u$, and the steps $\{ j_2, j_3 \} \to \{ j_3 \}$ and $\{ j_3, j_1 \} \to\{ j_1 \}$ imply $|j_2 \succ j_3| < t$ and $|j_3 \succ j_1| < t$. Notice that for any agent $i$ for which $j_2 \succ_i j_1$, we must have $j_2 \succ_i j_3$ or $j_3 \succ_i j_1$. Thus $|j_2 \succ j_1| \leq |j_2 \succ j_3| + |j_3 \succ j_1| \leq 2(t\!-\!1) < u$. 
\end{proof}

\begin{lemma}\label{lemma:12snake-len4}
Let $t,u \in \N$ such that $u \ge 2t\!-\!1$, and let $\problem$ be a menu selection problem with $g \ge 4$. For distinct $j_1, j_2, j_3, j_4 \in G$, if \cref{alg:greedy} cycles like 
\[ \{ j_1 \} \to \{ j_1, j_2 \} \to \{ j_2 \} \to \{ j_2, j_3 \} \to \{ j_3 \} \to \{ j_3, j_4 \} \to\{ j_4 \} \to \{ j_4, j_1 \} \to \{ j_1\}, \]
then $\{ j_1, j_3 \}$ and $\{ j_2, j_4 \}$ are $t$-feasible and $u$-uncontestable.
\end{lemma}
\begin{proof}
We will prove $\{ j_1, j_3 \}$ is $t$-feasible and $u$-uncontestable, as the proof is symmetric for $\{ j_2, j_4 \}$. To show $t$-feasibility, notice 
\begin{align*}
|j_3 \succ j_1| &\geq |j_3 \succ j_2| - |j_1  \succ j_2| \geq t \\
|j_1 \succ j_3| &\geq |j_1 \succ j_4| - |j_3 \succ j_4| \geq t,
\end{align*}
where here we are using $|j_3 \succ j_2|, |j_1 \succ j_4| \geq u$ and $|j_1  \succ j_2|, |j_3 \succ j_4| < t$. To show $u$-uncontestability, notice 
\begin{align*}
|j_2 \succ \{ j_1, j_3 \} | &\leq |j_2 \succ j_3| < t \leq u \\
|j_4 \succ \{ j_1, j_3 \} | &\leq |j_4 \succ j_1| < t \leq u,
\end{align*}
where here we are using $|j_2 \succ j_3|, |j_4 \succ j_1| < t$. 
\end{proof}

Finally, we combine the above lemmas to prove \cref{thm:positive-four}.

\begin{proof}[Proof of \cref{thm:positive-four}]
By \cref{lemma:suff-popular}, it suffices to consider the case $|j| < t$ for all $j \in G$.\footnote{Roughly speaking, this is because any goods $j$ with $|j| \ge t$ (that is, $j$ is the favorite good of at least $t$ agents) can immediately be opened. This reduces the original menu selection problem to a smaller menu selection problem such that each good is the favorite of $<t$ agents. See \cref{lemma:suff-popular} for details.} In this setting, every $O \subseteq G$ with $|O|=2$ is $u$-uncontestable, because the maximum possible lobby size is 
\[ \max_{j \in G \setminus O} |j \succ O| \leq (g\!-\!|O|)(t\!-\!1) \leq 2(t\!-\!1) < u.\]

Now, run \cref{alg:greedy} on this menu selection problem. If there is a cycle, the cycle must be of the form 
\[ \{ j_1 \} \to \{ j_1, j_2 \} \to \{ j_2 \} \to \{ j_2, j_3 \} \to \dots \to \{ j_{\ell-1} \} \to \{ j_{\ell - 1}, j_\ell \} \to \{ j_{\ell} \} \to \{ j_\ell, j_1 \} \to \{j_1 \},\]
where $\ell \leq 4$. By Lemmas \ref{lemma:12snake-len2} and \ref{lemma:12snake-len3}, in fact $\ell = 4$. And by Lemma \ref{lemma:12snake-len4}, a $t$-feasible and $u$-uncontestable $O\subseteq G$ exists.
\end{proof}

\subsection{Computer-Assisted Proof for $g=5,6$}\label{sec:computational-c56-proof}

In this section we prove the remaining two parts of \cref{upper-bound-2} from the introduction. That is, we prove the tightness of the lower bound of \cref{lower-bound-2} in the cases $g \in\{5,6\}$. 

\begin{theorem}\label{thm:positive-six}
For every $g \in \{5,6\}$ and $t,u \in \N$ such that $u \ge 2t\!-\!1$, every menu selection problem with $g$ public goods has a $(t,u)$-stable menu.
\end{theorem}

We prove this upper bound in a computer-aided manner: we analytically show that any instance of a menu selection problem is equivalent to a polyhedra covering problem, which we solve using the open-source SMT solver z3 \citep{de_moura_z3_2008}.\footnote{Our code is publicly available at the GitHub repository \url{https://github.com/sara-fish/stable-menus-of-public-goods}. There we explore two methods for solving our particular polyhedra covering problem: One, by encoding as a SMT problem (and solving with z3), and two, by encoding as an ILP (and solving with Gurobi). The SMT solver approach is substantially faster, so we focus on it in our exposition.}\textsuperscript{,}\footnote{We are of course far from the first to apply SMT solvers to economic settings. For example, SMT solvers have seen fruitful applications in social choice theory \citep[see, e.g.,][]{geist_computer-aided_2017,brandl_2018_proving,aziz_computational_2019}.} Specifically, we can encode any instance of a menu selection problem on $g$ public goods (for now assume complete preferences) as a point in $\Z^{g!}$, where each component counts the number of agents with a specific preference. Moreover, for each menu $O \subseteq G$, we show that the condition ``$O$ is $(t,u)$-stable for a menu selection problem'' is equivalent to the condition ``the menu selection problem's encoding $x \in \Z^{g!}$ is contained in some polyhedron $P_O^{g,t,u}$.'' Thus, every menu selection problem has a $(t,u)$-stable solution if and only if $\Z_{\geq 0}^{g!}$ is covered by the polyhedra $P_O^{g,t,u}$ for all $O \subseteq G$.

\subsubsection{Notation}

Fix the number of public goods $g$, as well as the parameters $t$ and $u$. Let $G = \{1,2,\dots, g \}$ denote the menu of all public goods. Let $\mc{P}$ denote some set of preferences $\succ$ over $G$, and fix some ordering of the elements of $\mc{P}$. 

First we specify how to convert a menu selection problem into a point in high-dimensional space. 

\begin{definition}
Let $\problem$ be a menu selection problem for which $\{ \succ_i:\ i \in N\} \subseteq \mc{P}$ for some set $\mc{P}$ of preferences over $G$. An \textit{agent cohort} is a vector $x \in \Z^{|\mc{P}|}$, indexed by preferences $\succ \in \mc{P}$, where $x_\succ$ is defined to be the number of agents in $\problem$ with preferences equal to $\succ$. 
\end{definition}

Next we define the polyhedra $P_O^{g,t,u}$ for all $O \subseteq G$ in such a way such that an agent cohort $x$ satisfies $x \in P_O^{g,t,u}$ if and only if $O$ is $(t,u)$-stable for the menu selection problem corresponding to $x$. 

\begin{definition}
The \textit{matrix associated with $O \subseteq G$} is denoted $A_O^{g,t,u}$. We specify $A_O^{g,t,u} \in \Z^{g \times |\mc{P}|}$ by specifying each row $a_i$ at a time. For all $i = 1, \dots, g$: 
\begin{itemize}
    \item ($t$-feasibility) If $i \in O$: for each $\succ \in \mc{P}$, set $a_i^{(\succ)} = 1$ if $i \succ O \setminus \{ i \}$, and 0 otherwise. 
    \item ($u$-uncontestability) If $i \notin O$: for each $\succ \in \mc{P}$, set $a_i^{(\succ)} = 1$ if $i \succ O$, and 0 otherwise. 
\end{itemize}
\end{definition}

\begin{definition}
The \textit{polyhedron associated with $O \subseteq G$}, denoted $P_O^{g,t,u}$, is given by $\{ x \in \Z^{|\mc{P}|}_{\geq 0}:\ D A_O^g x \geq b \}$, where $D \in \Z^{g \times g}$ and $b \in \Z^g$ are defined as follows:
\begin{itemize}
\item To specify $D$, set $D^{(ii)} := 1$ if $i \in O$ and $D^{(ii)} := -1$ if $i \notin O$ (and $D^{(ij)} := 0$ for $i \neq j$). 
\item To specify $b$, set $b^{(i)} := t$ if $i \in O$ and $b^{(i)} := -(u\!-\!1)$ if $i \notin O$.
\end{itemize}
\end{definition}

The following proposition follows directly from the above definitions. 

\begin{proposition}\label{thm:polyhedra-reduction}
Fix $t,u \in \N$. The following two statements are equivalent:
\begin{enumerate}[(1)]
    \item The polyhedra $P_O^{g,t,u}$ (for all $O \subseteq G$) cover $\Z^{|\mc{P}|}_{\geq 0}$, that is, $\Z^{|\mc{P}|}_{\geq 0} \subseteq \bigcup_{O \subseteq G} P_O^{g,t,u}$.
    \item For every menu selection problem on $g$ public goods, there exists a $(t,u)$-stable set. 
\end{enumerate}
\end{proposition}

\subsubsection{SMT Solver Calculation}

By \cref{thm:polyhedra-reduction}, to prove there exists a $(t,u)$-stable menu for all menu selection problems (over preferences $\mc{P}$) on $g$ public goods, it suffices to show $\Z^{|\mc{P}|}_{\geq 0} \subseteq  \bigcup_{O \subseteq G} P_O^{g,t,u}$. 

In other words, for $g \le 6$, we wish to show:
\[ \forall x \in \Z^{|\mc{P}|}_{\geq 0},\ \forall t,u \in \N \text{ s.t. } u \geq 2t\!-\!1,\ \exists O \subseteq G\text{ s.t. }x \in P_O^{g,t,u},\]
or equivalently,
\[ \neg\ [\exists x \in \Z^{|\mc{P}|}_{\geq 0},\ \exists t,u \in \N\text{ s.t. }u \geq 2t\!-\!1,\  \text{ s.t. }\forall O \subseteq G,\ x\notin P_O^{g,t,u}].\tag{$*$}\]
SMT solvers are especially well-equipped to verify expressions like $(*)$. An SMT (satisfiability modulo theories) solver takes as input simple mathematical and logical constraints over variables, and outputs whether the constraints are satisfiable or unsatisfiable. 

We obtain the most general result when we let $\mc{P}$ be the set of all preferences of the form $j_1 \succ \dots \succ j_{c'}$, where $g' \leq g$, and for $j_i$ distinct (that is, all preferences permitted in \cref{def:cap}). Note $|\mc{P}| = \sum_{k=0}^g \frac{g!}{k!}$. 

In our case, the variables consist of $|\mc{P}| + 2$ nonnegative integers (the $|\mc{P}|$ components of $x$, as well as $t$ and $u$), and there are $2^g$ constraints (one for each $O \subseteq G$). Recall $P_O^{g,t,u} = \{ x \in \Z_{\geq 0}^{|\mc{P}|}:\ DA_O^{g,t,u} x \geq b \}$, where $D$ is a diagonal matrix. Let $a_i$ denote the $i$th row of $A_O^{g,t,u} \in \Z^{g \times |\mc{P}|}$ and $b^{(i)}$ the $i$th entry of $b \in \Z^g$. Thus, 
\[ x \notin P_O^{g,t,u} \iff \bigvee_{i=1}^g [\langle D^{(ii)}a_i, x \rangle < b^{(i)}]. \]
This is an \textsf{OR} of $g$ integer linear inequalities, which are constraints SMT solvers can generally handle.

Using Z3Py, a Python API, we input ($*$) into the SMT solver z3. Specifically, we input the statement $\neg (*)$, so that an output of ``unsatisfiable'' means every menu selection problem has a $(t,u)$-stable set, for that particular value of $g$. For each $g=5,6$ (and also for $g=3,4$), the SMT solver returned ``unsatisfiable''.

Proving upper bounds using this method becomes computationally infeasible for larger~$g$. Note for example the exponential dependence of $|\mc{P}|$ on $g$. Moreover, fundamentally this method involves solving an instance of a polyhedra covering problem, and this family of problems is coNP-hard in general.\footnote{To encode an instance of 3-SAT on $n$ variables $x_1, \dots, x_n$ as a polyhedra covering problem (whether input polyhedra $P_i$ cover some polyhedron $Q$): set $Q = [0,1]^n$ and encode each clause as follows: e.g.\ clause $(x_i \vee x_j \vee \neg x_k)$ becomes the polyhedron $\{ x:\ x^{(i)} >1/2\}\cap \{x:\ x^{(j)} > 1/2\}\cap \{x:\ x^{(k)} \le 1/2 \}$. Finally, intersecting with $\Z^n$ does not change this analysis.} Nonwithstanding there is much structure to this particular instance (especially in the case of complete preferences, see \cref{app:matrix-structure-complete-case}). It is an interesting open problem to find a computationally efficient algorithm that checks, for fixed $g,t,u$, whether $\Z_{\geq 0}^{|\mc{P}|} \subseteq \bigcup_{O \subseteq G} P_O^{g,t,u}$. 

\section{\texorpdfstring{$g=7$}{g=7} and Beyond}\label{sec:beyondc6}

So far, taken together, \cref{lower-bound-2,thm:positive-four,thm:positive-six} provide a tight characterization, for $g \in \{3,4,5,6\}$, for the guaranteed existence of stable menus in terms of the ratio~$\frac{u-1}{t-1}$. (Specifically, existence is guaranteed if and only if $\frac{u-1}{t-1} \ge 2$.) One might conjecture that this pattern persists for $g \ge 7$. Surprisingly, this is not the case. 

Indeed, in \cref{sec:improved-lower-bound}, we show that a necessary condition for guaranteed existence of stable menus in the case $g \ge 7$ is $\frac{u-1}{t-1} \gtrapprox 2.1$. In \cref{sec:improved-upper-bound}, we further show that $\frac{u-1}{t-1} \ge g-2$ is a sufficient condition for guaranteed existence of stable menus in the case $g \ge 7$. Both of these results constitute improvements over the bounds presented in \cref{sec:simple-lower-upper-bounds} when $g \ge 7$. It remains an open question to identify a full characterization for the guaranteed existence of stable menus in terms of $\frac{u-1}{t-1}$ in the case $g \ge 7$ (see \cref{sec:discussion} for further discussion on this and other open questions). 

\subsection{Improved Lower Bound}\label{sec:improved-lower-bound}

Using an SMT solver (see \cref{sec:computational-c56-proof} for details), we computationally found a menu selection problem on $g\!=\!7$ public goods for which there exists no $(t,u)$-stable menu in the case $u \geq 2t\!-\!1$. The SMT solver outputted a complex menu selection problem on $n=144$ agents with $t=24$ and $u=47$ (these $t,u$ in particular satisfy $\frac{u-1}{t-1}=2$). By studying this counterexample and manually simplifying it, we obtain the following improved lower bound (cf. \cref{lower-bound-2}). 

\lowerboundcomplex*

We prove \cref{lower-bound-2311} by constructing a menu selection problem for which no $(t,u)$-stable solution exists. 

To construct this problem, let $x := \lfloor \frac{t-1}{11} \rfloor$. Set $u' := 23x$ and $t' := 11x + 1$. Then $u \leq u'$ and $t' \leq t$. By monotonicity (\cref{lemma:monotonicity}), it suffices to find a menu selection problem for which there exists no $(t', u')$-stable menu of public goods. To this end, consider the set of agents in \cref{tab:better-lower-bound}. 

\begin{table}
\small   
    \centering
    \begin{tabular}{c|c}
        $\#$ & Preference \\
        \hline 
        $5x$ & $1 \succ 2 \succ 3$ \\
        $5x$ & $2 \succ 3 \succ 4$ \\
        $5x$ & $3 \succ 4 \succ 5$ \\
        $5x$ & $4 \succ 5 \succ 6$ \\
        $5x$ & $5 \succ 6 \succ 7$ \\
        $5x$ & $6 \succ 7 \succ 1$ \\
        $5x$ & $7 \succ 1 \succ 2$ 
    \end{tabular}\quad~~
        \begin{tabular}{c|c}
        $\#$ & Preference \\
        \hline 
        $3x$ & $1 \succ 2 \succ 4 \succ 5$ \\
        $3x$ & $2 \succ 3 \succ 5 \succ 6$ \\
        $3x$ & $3 \succ 4 \succ 6 \succ 7$ \\
        $3x$ & $4 \succ 5 \succ 7 \succ 1$ \\
        $3x$ & $5 \succ 6 \succ 1 \succ 2$ \\
        $3x$ & $6 \succ 7 \succ 2 \succ 3$ \\
        $3x$ & $7 \succ 1 \succ 3 \succ 4$ 
    \end{tabular}\quad~~
        \begin{tabular}{c|c}
        $\#$ & Preference \\
        \hline 
        $x$ & $1 \succ 4 \succ 2 \succ 5$ \\
        $x$ & $2 \succ 5 \succ 3 \succ 6$ \\
        $x$ & $3 \succ 6 \succ 4 \succ 7$ \\
        $x$ & $4 \succ 7 \succ 5 \succ 1$ \\
        $x$ & $5 \succ 1 \succ 6 \succ 2$ \\
        $x$ & $6 \succ 2 \succ 7 \succ 3$ \\
        $x$ & $7 \succ 3 \succ 1 \succ 4$ 
    \end{tabular}\quad~~
        \begin{tabular}{c|c}
        $\#$ & Preference \\
        \hline 
        $x$ & $1 \succ 6 \succ 4 \succ 2$ \\
        $x$ & $2 \succ 7 \succ 5 \succ 3$ \\
        $x$ & $3 \succ 1 \succ 6 \succ 4$ \\
        $x$ & $4 \succ 2 \succ 7 \succ 5$ \\
        $x$ & $5 \succ 3 \succ 1 \succ 6$ \\
        $x$ & $6 \succ 4 \succ 2 \succ 7$ \\
        $x$ & $7 \succ 5 \succ 3 \succ 1$ 
    \end{tabular}
    \caption{A set of $70x$ agents, for which there exists no $(11x+1, 23x)$-stable menu of public goods. A similar construction in which all agents have complete preferences can be found in \cref{app:c7example-complete}.}
    \label{tab:better-lower-bound}
\end{table}

First, observe that any $t'$-feasible $O \subseteq G$ must satisfy $O \subseteq \{ 1, 2, 3, 4, 5, 6, 7 \}$, because these are the only ranked public goods. Thus, to show the set of agents in \cref{tab:better-lower-bound} has no $(t', u')$-stable menu of public goods, by monotonicity (\cref{lemma:monotonicity}) it suffices to prove the following two lemmas.

\begin{lemma}\label{lemma:o3}
Every $O \subseteq \{ 1, 2, 3, 4, 5, 6, 7 \}$ with $|O|=3$ is not $u'$-uncontestable. 
\end{lemma}
\begin{lemma}\label{lemma:o4}
Every $O \subseteq \{ 1, 2, 3, 4, 5, 6, 7 \}$ with $|O|=4$ is not $t'$-feasible. 
\end{lemma}

We verify these claims by casework. Note by symmetry, each of the cases $|O|=3$ and $|O|=4$ only require checking $\binom73 / 7 = \binom74 /7  = 5$ menus. 

\begin{proof}[Proof of \cref{lemma:o3}]
By symmetry, it suffices to show each of the $\binom73 / 7 = 5$ values of $O$ are not $u'$-uncontestable: 
\begin{itemize}
    \item $O = \{ 1, 2, 3 \}$. Then $|6 \succ O| = 5x + 5x + 5x + 3x + 3x + x + x = 23x \geq u'$.
    \item $O = \{ 1, 2, 4 \}$. Then $|7 \succ O| = 5x + 5x + 5x + 3x + 3x + x + x = 23x \geq u'$.
    \item $O = \{ 1, 2, 5 \}$. Then $|7 \succ O| = 5x + 5x + 3x + 3x + 3x + x + x + x + x = 23x \geq u'$.
    \item $O = \{ 1, 2, 6 \}$. Then $|5 \succ O| = 5x + 5x + 5x + 3x + 3x + x + x + x + x = 23x \geq u'$.
    \item $O = \{ 1, 3, 5 \}$. Then $|7 \succ O| = 5x + 5x + 3x + 3x + x + x + x + x + x + x + x = 23x \geq u'$.\qedhere
\end{itemize}
\end{proof}

\begin{proof}[Proof of \cref{lemma:o4}]
By symmetry, it suffices to show each of the $\binom74 / 7 = 5$ values of $O$ are not $t'$-feasible: 
\begin{itemize}
    \item $O = \{ 1, 2, 3, 4 \}$. Then $|4 \succ O \setminus \{ 4 \}| = 5x + 3x + x + x + x = 11x < t'$.
    \item $O = \{ 1, 2, 3, 5 \}$. Then $|3 \succ O \setminus \{ 3 \}| = 5x + 3x + x + x + x = 11x < t'$.
    \item $O = \{ 1, 2, 3, 6 \}$. Then $|2 \succ O \setminus \{ 2 \}| = 5x + 3x + x + x + x = 11x < t'$.
    \item $O = \{ 1, 2, 4, 5 \}$. Then $|5 \succ O \setminus \{ 5 \}| = 5x + 3x + x + x + x = 11x < t'$.
    \item $O = \{ 1, 2, 4, 6 \}$. Then $|2 \succ O \setminus \{ 2 \}| = 5x + 3x + x + x = 10x < t'$.\qedhere 
\end{itemize}
\end{proof}

\begin{remark}
Using an SMT solver (see \cref{sec:computational-c56-proof}), we attempted to search for stronger lower bounds by checking whether $11u \ge 23(t\!-\!1)+1$ is satisfiable when $g \ge 7$. One month of computation on the equivalent of a modern laptop were not sufficient to determine satisfiability.
\end{remark}

\subsection{Improved Upper Bound}\label{sec:improved-upper-bound}

In this section we slightly improve upon the upper bound in \cref{upper-bound-g-minus-1}. Unlike the prior results described in this paper, we prove the following result only for the case of menu selection problems on complete preferences (that is, each agent has a complete preference list over all goods in $G$).\footnote{The proof of \cref{upper-bound-g-minus-2-formal} requires complete preferences because it relies on \cref{lemma}, which does not hold for the case of incomplete preferences, and we do not know how to remove the dependency on \cref{lemma} from the proof. In \cref{app:structural-upper}, we (using conceptually distinct techniques) prove an upper bound that holds for general incomplete preferences and is an improvement over \cref{upper-bound-g-minus-1}, but is a weaker result than \cref{upper-bound-g-minus-2-formal} for the case of complete preferences.}

\begin{theorem}\label{upper-bound-g-minus-2-formal}
For every $g\ge 7$ and $t,u \in \N$ such that $\frac{u-1}{t-1} \geq g\!-\!2$, every menu selection problem on complete preferences with $g$ goods has a $(t,u)$-stable menu.
\end{theorem}

At the heart of \cref{upper-bound-g-minus-2-formal} is the following lemma. 

\begin{lemma}\label{lemma:no-12-gaps}
Fix any $g \geq 2$ and $t,u \in \N$ such that $u \geq 2t\!-\!1$. Then for every menu selection problem on complete preferences, at least one of the following is true: 
\begin{enumerate}[(1)]
\item There exists some $O \subseteq G$ with $|O|=1$ such that $O$ is $u$-uncontestable. 
\item There exists some $O \subseteq G$ with $|O|=2$ such that $O$ is $t$-feasible. 
\end{enumerate}
\end{lemma}

\begin{remark}
In both of our lower bound constructions (\cref{lower-bound-2}, \cref{lower-bound-2311}), the constructed menu selection problem not only satisfies the property that for all $O \subseteq G$, either $O$ is not $t$-feasible or $O$ is not $u$-uncontestable. In fact, a stronger property holds: there exists some $k$, such that all $O \subseteq G$ with $|O|=k$ are not $u$-uncontestable, and all $O \subseteq G$ with $|O|=k+1$ are not $t$-feasible. (In the case of \cref{lower-bound-2}, we have $k=1$, and in the case of \cref{lower-bound-2311}, we have $k=3$.) If there does not exist a menu selection problem satisfying this property, we say that there are \textit{no $(k,k+1)$-gaps} for parameters $g,t,u$. 

In particular, \cref{lemma:no-12-gaps} says that when $u \geq 2t\!-\!1$, there are no (1,2)-gaps. Moreover, \cref{lower-bound-2} demonstrates the existence of (1,2)-gaps when $u \leq 2t-2$. Thus, by \cref{thm:polyhedra-reduction}, when $g \leq 6$, fixing $t$ and $u$, we have that there are no (1,2)-gaps for $g,t,u$ if and only if all menu selection problems have a $(t,u)$-stable solution. 

For $g \geq 7$, it is an open question whether such ``no gaps'' conditions are sufficient to imply the existence of $(t,u)$-stable solutions (it follows from monotonicity that they are necessary). 
\end{remark}

We first prove \cref{upper-bound-g-minus-2-formal} using \cref{lemma:no-12-gaps} as well as two simple lemmas that we prove in \cref{app:lemmas}, and later prove \cref{lemma:no-12-gaps}.

\begin{proof}[Proof of \cref{upper-bound-g-minus-2-formal}]
By \cref{lemma:no-12-gaps}, it suffices to consider two cases. 

In the first case, there exists some $u$-uncontestable $O \subseteq G'$ with $|O|=1$. Moreover, $O$ is $t$-feasible, because by \cref{lemma:suff-rarely-ranked}, we can assume that each public good is ranked by at least $t$ agents. Thus $O$ is $(t,u)$-stable. 

In the second case, there exists some $t$-feasible $O \subseteq G'$ with $|O|=2$. By \cref{lemma:suff-popular}, we can assume $|j| < t$ for each $j \in G$, so the maximum possible lobby size is 
\[ \max_{j \in G' \setminus O} |j \succ O| \leq (|G'| - |O|)(t\!-\!1) \leq (g\!-\!2)(t\!-\!1) < u.\]
Thus $O$ is $(t,u)$-stable. 
\end{proof}

In the case of $g \ge 7$, we thus obtain an improved upper bound $\frac{u-1}{t-1} \ge g-2$ (compared to \cref{upper-bound-g-minus-1}, which yields $\frac{u-1}{t-1} \ge g-1$). For $g \in \{2,3,4,5,6\}$, this argument still goes through, but yields worse bounds than the analytic proof via a greedy algorithm for $g \in \{2,3,4 \}$ (\cref{thm:positive-four}) and the SMT-solver based proof for $g \in \{5,6\}$ (\cref{thm:positive-six}).

\subsubsection{Proof of \cref{lemma:no-12-gaps}}

\begin{proof}
Case $g=2$: suppose neither (1) nor (2) holds. Then $|1 \succ 2| \geq u$, $|2 \succ 1| \geq u$, and $|1 \succ 2| < t\textsf{ OR }|2 \succ 1| < t$. Since $u \geq t$ these cannot all hold.

Now we consider some $g \geq 3$ (and assume the claim holds for smaller $g$). Again assume neither (1) nor (2) hold. We aim to derive a contradiction. Write $x_{ij} := |i \succ j|$ for short. Then all of the following statements hold: 
\begin{itemize}
\item[($*$)] \textit{All $O$ with $|O|=1$ not $u$-uncontestable: }For all $i \in G$: $\bigvee_{j \neq i} (x_{ji} \geq u)$.
\item[($**$)] \textit{All $O$ with $|O|=2$ not $t$-feasible: }For all $\{i,j\} \subseteq G$: $(x_{ij} < t) \vee (x_{ji} < t)$. 
\end{itemize}
Let $B_k := G \setminus \{k\}$. Consider the following smaller instance of the problem: 
\begin{itemize}
\item For all $i \in B$: $\bigvee_{j \neq i} (x_{ji} \geq u)$.
\item For all $\{i,j\} \subseteq B$: $(x_{ij} < t) \vee (x_{ji} < t)$. 
\end{itemize}
By $(**)$, all of the $<t$ constraints here are true. By induction (this is a smaller subproblem), one of the $\geq u$ constraints must be false. Specifically, say the constraint $\bigvee_{j \neq i_k} (x_{ji_k} \geq u)$ is the false one, for some $i_k \in G$. But also, we know all of the $(*)$ constraints are true. Thus $x_{k i_k} \geq u$. Do this for all $k$. 

Also note: since $x_{k i_k} \geq u$ for all $k \in G$, by $(**)$, we have $x_{i_k k} < t$ for all $k$. 

Now, without loss of generality, we can relabel so that $i_1 = 2, i_2 = 3, \dots, i_\ell = 1$ for some $\ell \leq g$. Thus we end up with the following constraints:
\begin{align*}
&x_{12} \geq u,\ x_{23} \geq u,\ \dots,\ x_{\ell-1,\ell} \geq u,\ x_{\ell 1} \geq u \\
&x_{21} < t,\ x_{32} < t,\ \dots,\ x_{\ell, \ell-1} < t,\ x_{1 \ell} < t. 
\end{align*}

Now we will use the following lemma many times: 
\begin{lemma}\label{lemma}
For all distinct $i,j,k$, when $u \geq 2t\!-\!1$, if $x_{ij} \geq u$ and $x_{kj} < t$, then $x_{ik} \geq t$.
\end{lemma}
\begin{proof}
$|i \succ j| \geq u$ and $|k \succ j| < t$ means that there are $u$ agents who prefer $i$ over $j$, and only $<t$ of them prefer $k$ over $j$. This means that the rest prefer $j$ over $k$ (this is the step that invokes the assumption of complete preferences). Thus there exist at least $u - (t-1) \ge t$ agents with who prefer both $i$ over $j$ and $j$ over $k$. In particular, there exist $\ge t$ agents who prefer $i$ over $k$.
\end{proof}

Finally we write down a table of deductions, until we reach a contradiction. (The table's shorthand can be understood as follows: a row `$y_1 y_2 \dots y_\ell \mid \ge u$' means $y_i \ge u$ for all $i \in [\ell]$, and a row `$y_1 y_2 \dots y_\ell \mid <t$' means $y_i < t$ for all $i \in [\ell]$.)  
\begin{center}
\begin{tabular}{c|ccccc|c|c}
Row \# & & & & & & & Reason \\
\hline 
1 & $x_{12}$ & $x_{23}$ & $\dots$ & $x_{\ell-1,\ell}$ & $x_{\ell 1}$ & $\geq u$ & Given \\
\rowcolor{lightgray} 2 & $x_{21}$ & $x_{32}$ & $\dots$ & $x_{\ell, \ell-1}$ & $x_{1 \ell}$ & $< t$ & $(**)$ on \#1 \\
3 & $x_{32}$ & $x_{43}$ & $\dots$ & $x_{1, \ell}$ & $x_{2 1}$ & $< t$ & rearrange \\
4 & $x_{13}$ & $x_{24}$ & $\dots$ & $x_{\ell-1,1}$ & $x_{\ell 2}$ & $\geq t$ & Lemma \ref{lemma} on \#1, \#3 \\
\rowcolor{lightgray} 5 & $x_{31}$ & $x_{42}$ & $\dots$ & $x_{1, \ell-1}$ & $x_{ 2 \ell}$ & $< t$ & $(**)$ on \#4 \\
6 & $x_{42}$ & $x_{53}$ & $\dots$ & $x_{2 \ell}$ & $x_{ 3 1}$ & $< t$ & rearrange \\
7 & $x_{14}$ & $x_{25}$ & $\dots$ & $x_{\ell-1,2}$ & $x_{\ell 3}$ & $\geq t$ & Lemma \ref{lemma} on \#1,\#6 \\ 
\rowcolor{lightgray} 8 & $x_{41}$ & $x_{52}$ & $\dots$ & $x_{2, \ell-1}$ & $x_{3 \ell}$ & $< t$ & $(**)$ on \# 7 \\ 
$\vdots$ & $\vdots$ & $\vdots$ & $\vdots$ & $\vdots$ & $\vdots$ & $\vdots$ & $\vdots$ \\
\rowcolor{lightgray} $3\ell-4$ & $x_{\ell 1}$ & $x_{12}$ & \dots & $x_{\ell-2,\ell-1}$ & $x_{\ell-1,\ell}$ & $<t$ & as above 
\end{tabular}
\end{center}
However, rows $1$ and $3\ell-4$ conflict. Thus, the constraints $x_{k i_k} \geq u$ are not simultaneously satisfiable, which in turn implies that the original $(*)$ and $(**)$ are not simultaneously satisfiable.
\end{proof}

\section{Strategic Considerations}\label{sec:strategic}

So far this paper has focused on the question of \textit{guaranteed existence} of solutions to the menu selection problem: for which $g,t,u$ does there exist a $(t,u)$-stable solution to every menu selection problem on $g$ public goods? In this section, we focus on parameters $g,t,u$ such that this existence question is settled in the affirmative,\footnote{When $g = 2$, we have guaranteed existence if and only if $u \ge t$. When $g \in \{3,4,5,6\}$, we have guaranteed existence if and only if $u \ge 2t-1$.} and instead consider the question of \textit{strategyproofness}. Specifically, we ask: for which $g$ does there exist a strategyproof stable menu selection mechanism?\footnote{For readers familiar with the work of \citet{sonmez_strategy-proofness_1999} on strategyproof selection from the core of generalized indivisible goods allocation problem, we note that the core correspondence of our menu selection problems, unlike in the problems that paper studies, is not \emph{essentially single valued}. Specifically, fix $u,t \in \mathbb{N}$ with $u \ge 2t-1$. Consider a menu selection problem on $g \in \{2,3,4,5,6\}$ goods in which $t$ agents have preferences $1 \succ 2 \succ *$ and $u$ agents have preferences $2 \succ *$ (where $*$ denotes arbitrary preferences on the remaining goods). Then $\{2\}$ and $\{1,2 \}$ are both $(t,u)$-stable, but $\{1,2\}$ Pareto dominates~$\{2\}$. Thus, the core correspondence of our menu selection problems is not essentially single-valued.} We start by defining what we mean by stable menu selection mechanism and by strategyproofness of such a mechanism.

\begin{definition}
Fix $g \ge 2$ and let $\mc{P}$ be a set of menu selection problems on $g$ public goods. A set of stability parameters $\mc{S} \subseteq \mathbb{N}^2$ \textit{guarantees existence} for $\mc{P}$ if for all $(t, u) \in \mc{S}$, every menu selection problem $P \in \mc{P}$ has a $(t,u)$-stable solution. 
\end{definition}

For example, when $g=2$, the set $\mc{S} = \{(t,u):\ u \ge t \}$ guarantees existence for all menu selection problems, and when $g \in \{3,4,5,6\}$, the set $\mc{S} = \{(t,u):\ u \ge 2t-1 \}$ guarantees existence for all menu selection problems. 

\begin{definition}[Stable Menu Selection Mechanism]
Fix $g \ge 2$ and let $\mc{P}$ be a set of menu selection problems on $g$ public goods. Let $\mc{S} \subseteq \N^2$ be a set of stability parameters that guarantees existence for $\mc{P}$. A mechanism $\mc{M} : \mc{P}\times\mc{S} \to 2^G$ is said to be \textit{stable} if for every $P \in \mc{P}$ and $(t,u)\in\mc{S}$, the menu $\mc{M}\bigl(P,(t,u)\bigr)$ is $(t,u)$-stable. 
\end{definition}

\begin{definition}[Strategyproof Menu Selection Mechanism]
Fix $g \ge 2$ and let $\mc{P}$ be a set of menu selection problems on $g$ public goods. Let $\mc{S} \subseteq \N^2$ be a set of stability parameters that guarantees existence for $\mc{P}$. A mechanism $\mc{M} : \mc{P}\times\mc{S} \to 2^G$ is said to be \textit{strategyproof} if for all agents $i \in N$, for all menu selection problems $\problem = P \in \mc{P}$ and $\bigl(N, G, (\succ_i')_{i \in N}\bigr) = P' \in \mc{P}$ with $\succ_j = \succ_j'$ when $j \neq i$, we have 
\begin{equation*} a_{\mc{M}(P, (t,u))}(i)  \succcurlyeq_i  a_{\mc{M}(P', (t,u))}(i). \end{equation*}
That is, no agent can strictly improve their outcome by misreporting. 
\end{definition}

\begin{definition}[Anonymous Menu Selection Mechanism]
Fix $g \ge 2$ and let $\mc{P}$ be a set of menu selection problems on $g$ public goods. Let $\mc{S} \subseteq \N^2$ be a set of stability parameters that guarantees existence for $\mc{P}$. A mechanism $\mc{M} : \mc{P} \times \mc{S} \to 2^G$ is said to be \textit{anonymous} if for every combination of menu selection problem $\problem \in \mc{P}$, stability parameters $(t,u) \in \mc{S}$, and permutation $\sigma : N \to N$, we have 
\begin{equation*} M\Bigl( \problem, (t,u)\Bigr) = M\Bigl( \bigl(N, G, (\succ_{\sigma(i)})_{i \in N} \bigr) , (t,u)\Bigr). \end{equation*}
That is, each agent is treated equally by the mechanism, in the sense that the output of the mechanism does not depend on the identity of the agent who submitted each preference list.
\end{definition}

In this section, we ask for which numbers of goods $g$ it is the case that strategyproof anonymous stable mechanisms exist. In the case $g=2$, we exhibit a simple anonymous mechanism on complete preferences and prove it is stable and strategyproof (\cref{sec:c2-strategyproof}). In the case $g \in \{3,4,5,6\}$, we prove fundamental barriers for strategyproofness that contrast with the case of $g=2$. 

\subsection{Strategyproof Mechanism for $g=2$}\label{sec:c2-strategyproof}

Below, we exhibit a simple, natural strategyproof and anonymous mechanism that produces stable solutions to the menu selection problem for the case of complete preferences and $g = 2$. 

\begin{proposition}[Strategyproofness for $g=2$ on Complete Preferences]\label{sp-2}
Fix $g=2$. Let $\mc{S}=\bigl\{(t,u)\in\N^2~\big|~u \geq t\bigr\}$. Let $\mc{P}$ be the set of all menu selection problems on complete preferences with two public goods. Define a mechanism $\mc{M}: \mc{P}\times\mc{S} \to 2^G$ such that for every $(t,u)\in\mc{S}$ and $P\in\mc{P}$:
\begin{itemize}
    \item If $P$ has less than $t$ agents, then $\mc{M}\bigl(P,(t,u)\bigr):=\emptyset$. \textit{(``offer $\emptyset$ if no other option'')}
    \item Otherwise, if $\{1,2 \}$ is $(t,u)$-stable for $P$, then $\mc{M}\bigl(P,(t,u)\bigr):=\{ 1,2 \}$. \textit{(``offer $\{1,2\}$ if possible'')}
    \item Otherwise, $\mc{M}\bigl(P,(t,u)\bigr):=\{1\}$ if $|1 \succ 2| \geq |2 \succ 1|$ and $\mc{M}\bigl(P,(t,u)\bigr):=\{2\}$ otherwise. \textit{(``majority vote'')}
\end{itemize}
Then $\mc{M}$ is a strategyproof anonymous stable mechanism. 
\end{proposition}
\begin{proof}
If $P$ has less than $t$ agents, then no manipulation can affect this condition, so $\emptyset$ will be offered no matter what. Thus we can assume $n \ge t$. 

Because $P$ is a menu selection problem on complete preferences, each agent either has the preference list $1 \succ 2$ or $2 \succ 1$. Set $x_1 := |1 \succ 2|$ and $x_2 := |2 \succ 1|$. Observe the output of $\mc{M}(\cdot, (t,u))$ only depends on $x_1$ and $x_2$ (the condition ``$\{1,2\}$ is $(t,u)$-stable'' is equivalent to ``$x_1 \ge t$ and $x_2 \ge t$''). We can thus define a function $M : \Z_{\ge 0}^2 \to 2^G$ given by $M(x_1, x_2) := \mc{M}(P, (t,u))$, where $P$ is any menu selection problem such that $x_1 := |1 \succ 2|$ and $x_2 := |2 \succ 1|$. We call such an $\mathbf{x}$ the \textit{summary statistics} of a given set of preferences. 

Next we observe that $M$ satisfies \textit{monotonicity} in the following sense: for all $\mathbf{x} \in \Z_{\geq 0}^2$ and $j \in \{1,2\}$, we have $M(\mathbf{x} + e_j)\setminus M(\mathbf{x}) \subseteq \{ j\}$, where $e_1 = (1,-1)$ and $e_2 = (-1,1)$.  

Now suppose there exists a profitable manipulation by an agent with preferences $j \succ k$. If truthful reporting corresponds to summary statistics $\mathbf{x}$, then this agent's manipulation would result in summary statistics $\mathbf{x} + e_k$. Because $n \ge t$, such a manipulation must cause this agent's top choice $j$ to be offered, when it was previously not offered. But by monotonicity, if $j \notin M(\mathbf{x})$, then $j \notin M(\mathbf{x} + e_k)$, so such a manipulation does not exist.
\end{proof}

\subsection{Impossibility of Stategyproofness for $3 \le g \le 6$}\label{sec:c3-impossibility}

Next, we show that for $g \in \{3,4,5,6 \}$, any anonymous mechanism that outputs stable solutions to the menu selection problem cannot also be strategyproof. 

\begin{theorem}\label{sp-3456}
Fix $g \in \{3,4,5,6\}$. Let $\mc{S}=\bigl\{(t,u)\in\N^2~\big|~u \ge 2t-1\bigr\}$. Let $\mc{P}$ be the set of all menu selection problems on complete preferences with $g$ public goods. Every anonymous stable mechanism $\mc{M} : \mc{P} \times\mc{S} \to 2^G$ is not strategyproof. 
\end{theorem}
\begin{proof}
Let $\mc{M}$ be an anonymous stable mechanism. Let $t\in\N$, let $u=2t-1$, and let $\mc{P}_u\subseteq\mc{P}$ be the set of all menu selection problems in $\mc{P}$ with $u$ agents. 

First, we claim that $\mc{M}\bigl(P,(t,u)\bigr)$ is a singleton set for all $P \in \mc{P}_u$. Because $P$ has $u$ agents and the agents' preferences are complete, $\emptyset$ is not $(t,u)$-stable, so $\bigl|\mc{M}\bigl(P,(t,u)\bigr)\bigr| \ge 1$. Additionally, if some $O \subseteq G$ with $|O| \ge 2$ were $t$-feasible, then there would be at least $2t$ agents. But there are only $u = 2t\!-\!1$ agents. Thus $\bigl|\mc{M}\bigl(P,(t,u)\bigr)\bigr| \le 1$. 

Because $\mc{M}\bigl(\cdot,(t,u)\bigr)$ always outputs a singleton set over $\mc{P}_u$, it induces a voting rule on $g$ alternatives and $u$ voters. Anonymity of $\mc{M}$ implies that this voting rule is nondictatorial. Moreover, because $\mc{M}$ is stable, this voting rule is unanimous. To see this, consider any problem $P\in\mc{P}_u$ in which some public good $j \in G$ is top ranked by all agents. Then $O = \{ j \}$ is the unique $(t,u)$-stable solution, because $|k \succ j| = 0$ for every public good $k \neq j$. So $\mc{M}\bigl(P,(t,u)\bigr) = \{j\}$. 

By the Gibbard-Satterthwaite Theorem \citep{gibbard_manipulation_1973,satterthwaite_strategy-proofness_1975}, a unanimous nondictatorial voting rule is manipulable. Therefore, the voting rule induced by $\mc{M}$ is manipulable, and hence so is $\mc{M}$. Therefore, $\mc{M}$ is not strategyproof.
\end{proof}

\section{Discussion}\label{sec:discussion}

In this paper we formulated and initiated the study of a matching problem of a new flavor, provided nontrivial upper and lower bounds on the existence of a solution, and investigated the possibility of a strategyproof solution while unearthing a connection to voting theory. The problem, as formulated, is rather stylized. Yet, some other matching problems that started out similarly stylized eventually found real-world applications. 

A main limitation of our results is that they only give practical unrestricted guarantees for small numbers of public goods. For example, an online university looking to select which subset of six electives to offer to students in a specific small department can make use of the fact that a $(11,21)$-stable menu exists. But, if that same university instead aims to select which subset of 600 electives to offer to the entire student body, then absent further assumptions, we only give a considerably less useful guarantee of the existence of a \mbox{$(11, 5981)$-stable} menu. Under added assumptions, however, our results imply useful guarantees even in large-scale settings. For example, suppose the university offering 600 electives has 100 departments, each of which aims to select a subset of six electives to offer, taking in account the preferences of students only in that department (and students most prefer to take electives within their department). Then, due to the ``block structure'' of the students' preferences, we guarantee the existence of a $(11,21)$-stable menu despite the large scale of this scenario.

For the case of two public goods, our strategyproof mechanism is natural and practical: offer both if feasible; otherwise offer the most popular option. For between three and six public goods, we show that so long as $u \ge 2t\!-\!1$, a stable menu of public goods exists, but it is generally impossible to find it in a strategyproof manner. 

We write ``generally impossible'' because there might still be opportunity for strategyproof matching in special cases. Our proof of \cref{sp-3456} (corresponding to \cref{sp-3456-informal} from the introduction) utilizes a special case with $n=u=2t\!-\!1$ to derive the impossibility.\footnote{In a nutshell, $u=2t\!-\!1$ is used to guarantee singleton menus, and hence a single winner of the voting rule, which is needed in order to apply the Gibbard--Satterthwaite Theorem. One might imagine relaxing this constrainst by appealing to variants of the Gibbard--Satterthwaite Theorem that allow for nonsingleton sets of winners, such as the Duggan–-Schwartz Theorem \citep{duggan_strategic_2000}. Under suitable assumptions this would lead to the conclusion that the voting rule (and hence the menu selection mechanism) is manipulable, however, the notion of manipulability that is guaranteed by the Duggan--Schwartz Theorem differs from ours, and crucially, counts as manipulations some things that we do not consider as such in the menu selection setting. For example, if one's true preferences are $2\succ1$, and $\{1,2\}$ is offered, in their setting misreporting so that only $\{2\}$ is offered is considered a profitable manipulation, while in our setting we do not consider it as such.} In many real-life settings, it might suffice to consider $u>2t\!-\!1$ (or more generally, a smaller set $\mc{S}$ of supported stability parameters) or $n>u$ (or more generally, a smaller set $\mc{P}$ of supported menu selection problems), and future research might wish to more finely map out the landscape of possible $\mc{S}$ and $\mc{P}$ for which strategyproof matching might be possible, including for higher values of $g$.

A potential barrier for usability from which our analysis does not suffer is computational tractability. Indeed, all our upper bounds are computationally tractable in the sense that a stable solution can be feasibly found. Our main open question from the introduction asks for tighter bounds for guaranteed existence of a stable solution, and it is plausible that tighter upper bounds obtained in follow-up work might be harder (or impossible) to translate into tractable algorithms. An open question that we already have in a similar vain is on the precise asymptotic running time for our tighter upper bound from \cref{upper-bound-g-minus-2-formal} (corresponding to \cref{upper-bound-g-minus-2} from the introduction). The algorithm that follows from our proof of that theorem is quadratic in the number of public goods, while the algorithm that follows from the proof of the looser upper bound from \cref{upper-bound-g-minus-1} is linear. In \cref{app:structural-upper} we prove an upper bound slightly weaker than that from \cref{upper-bound-g-minus-2-formal} but nonetheless stronger than that from \cref{upper-bound-g-minus-1}, which still yields a linear-time algorithm. Is there a linear-time algorithm for finding a stable menu for all~$g,t,u$ such that $\frac{u-1}{t-1}\ge g\!-\!2$?

A final direction for generalization is moving from complete preferences to incomplete preferences. Many of our proofs work also for incomplete preferences, but for some it remains to be seen whether they can be extended.

\printbibliography[heading=bibintoc]

\addtocontents{toc}{\vspace{1em}}

\appendix 

\addtocontents{toc}{\protect\setcounter{tocdepth}{1}}

\section{Supporting Lemmas}\label{app:lemmas}

Below we write $\sqcup$ to denote disjoint union. 

\begin{definition}\label{def:complete-embedding}
Let $\problem$ be a menu selection problem on $g$ public goods (with possibly incomplete preferences). Its \textit{complete embedding} is defined to be the following menu selection problem $\anotherproblem$: 
\begin{itemize}
    \item Let $G' := G \sqcup \{g + 1 \}$. 
    \item Let $N' := N \sqcup \{ n + 1, \dots, n + u \}$. 
    \item Pick $i \in N'$. We specify agent $i$'s preferences as follows: 
    \begin{itemize}
        \item Case 1: $i \in N$. Suppose agent $i$ has preference $j_1 \succ_i \dots \succ_i j_{g'}$ in $\problem$, where $g' \leq g$. Then we define agent $i$'s preferences in $\anotherproblem$ as follows: $j_1 \succ_i' \dots \succ_i' j_{g'} \succ_i' g\!+\!1 \succ_i' \text{(arbitrary ranking of remaining public goods)}$.
    \item Case 2: $i \notin N$. In this case we define agent $i$'s preferences in $\anotherproblem$ as follows: $g\!+\!1 \succ_i' \text{(arbitrary ranking of remaining public goods)}$. 
    \end{itemize}
\end{itemize}
\end{definition}

\begin{lemma}[Reduction from incomplete preferences on $g$ public goods to complete preferences on $g\!+\!1$ public goods]\label{lemma:suff-complete}
Fix $u \geq t$. Let $\problem$ be a menu selection problem and let $\anotherproblem$ be its complete embedding. Then we have the following:
\begin{enumerate}[(1)]
    \item If $O \subseteq G'$ is $(t,u)$-stable for $\anotherproblem$, then $g\!+\!1 \in O$. 
    \item For each $O \subseteq G$, we have that $O$ is $(t,u)$-stable for $\problem$ if and only if $O \cup \{ g\!+\!1 \}$ is $(t,u)$-stable for $\anotherproblem$. 
    \item There exists a $(t,u)$-stable solution for $\problem$ if and only if there exists a $(t,u)$-stable solution for its complete embedding $\anotherproblem$. 
    \end{enumerate}

\end{lemma}

\begin{proof}
By construction, there are at least $u$ agents whose favorite public good is $g\!+\!1$. This implies (1). Moreover, (3) is directly implied by (2). Finally, (2) is implied by the following: 
\begin{equation*}
    |k \succ O\setminus \{ k \}| = |k \succ' (O \sqcup \{g\!+\!1\}) \setminus \{ k \}|\text{ for all }k \in G,\ O \subseteq G.\qedhere 
\end{equation*}\end{proof}

\begin{lemma}[Removing rarely ranked public goods]\label{lemma:suff-rarely-ranked}
Fix $u \geq t$ and let $\problem$ be a menu selection problem. Let $G' \subseteq G$ be the set of all public goods which are ranked by at least $t$ agents. Now define $\subproblem$ to be another menu selection problem, in which each agent $i$ has the same preferences as in $\problem$, but with any public goods in $G \setminus G'$ removed from their ranking. Then $O \subseteq G$ is $(t,u)$-stable for $\problem$ if and only if $O \subseteq G'$ and $O \subseteq G'$ is $(t,u)$-stable for $\subproblem$. 
\end{lemma}
\begin{proof}
The main idea is to observe the following fact:
\begin{equation}\label{eq:reduction-rarely}
|k \succ O \setminus \{ k \}| = |k \subsucc O \setminus \{ k \}|\text{ for all }k \in G',\ O \subseteq G'.
\end{equation}
Suppose $O \subseteq G$ is $(t,u)$-stable for $\problem$. Then $O \subseteq G'$, because any $k \in G \setminus G'$ is ranked by agents in $\problem$ less than $t$ times, so certainly $|k \succ O \setminus \{ k \}| < t$. Then \cref{eq:reduction-rarely} implies that $O \subseteq G'$ is $(t,u)$-stable for $\subproblem$. 

Now suppose $O \subseteq G'$ is $(t,u)$-stable for $\subproblem$. By \cref{eq:reduction-rarely}, $O$ is $t$-feasible for $\problem$, and for $u$-uncontestability, it remains to check that $|k \succ O| < u$ for all $k \in G \setminus G'$. This holds because each $k \in G \setminus G'$ is ranked $< t$ times, so $|k \succ O | < t \leq u$.
\end{proof}

\begin{lemma}[Offering sufficiently popular public goods]\label{lemma:suff-popular}
Fix $u \geq t$ and let $\problem$ be a menu selection problem. Let $G' \subseteq G$ be the set of all public goods $k$ for which there exist (strictly) fewer than $t$ agents with favorite public good $k$. Now define $\subproblem$ to be another menu selection problem, in which each agent $i$ has the same preferences as in $\problem$, but with any public goods in or ranked below $G \setminus G'$ removed from their ranking. Then $O \sqcup (G \setminus G') \subseteq G$ is $(t,u)$-stable for $\problem$ if and only if $O \subseteq G'$ is $(t,u)$-stable for $\subproblem $. 
\end{lemma}
\begin{proof}
The main idea is to observe the following fact: 
\begin{equation}\label{eq:reduction-popular}
|k \succ (O \sqcup (G \setminus G')) \setminus \{ k \} | = |k \subsucc O \setminus \{ k \}|\text{ for all }k \in G',\ O \subseteq G'.
\end{equation}
The $\Rightarrow$ direction follows directly from \cref{eq:reduction-popular}. Now suppose $O \subseteq G'$ is $(t,u)$-stable for $\subproblem$. We wish to show $O \sqcup (G \setminus G')$ is $(t,u)$-stable for $\problem$. By \cref{eq:reduction-popular}, $O \sqcup (G \setminus G')$ is $u$-uncontestable for $\problem$, and for $t$-feasibility, it remains to check $|j \succ (O \sqcup (G \setminus G')) \setminus \{ j \}| \geq t$ for all $j \in G \setminus G'$. This holds because $|j| \geq t$ for every  $j \in G \setminus G'$, by assumption. 
\end{proof}

\section{Inadequacy of a Greedy Algorithm for Finding a Stable Menu}\label{app:greedy}

In \cref{thm:positive-four}, we prove that a $(t,u)$-stable menu of public goods exists for all menu selection problems with $u \geq 2t\!-\!1$ in the case $g=4$ by showing that if the greedy algorithm (\cref{alg:greedy}) cycles, then we can study the cycle transcript and find a $(t,u)$-stable $O \subseteq G$ directly. 

Unfortunately, this approach does not immediately generalize to larger $g$. In this setting, the greedy algorithm (\cref{alg:greedy}) may cycle, and when it does, it is not possible to determine a $(t,u)$-stable menu from the cycle transcript alone. 

\begin{example}[Cycle transcript alone does not determine menu]
Consider the following two instances of menu selection problems, for $g = 5$, $t = 4$, and $u = 7$: 
\begin{center}
\begin{tabular}{c|c}
Agent Cohort A & Agent Cohort B \\
$t=4, u=7, n=15$ & $t=4, u=7, n=15$ \\
\hline 
\begin{tabular}{cc}
Agent type & Quantity \\
\hline
$5 \succ 4 \succ 3 \succ 2 \succ 1$ & 2 \\
$4 \succ 3 \succ 2 \succ 1 \succ 5$ & 2 \\
$3 \succ 2 \succ 1 \succ 5 \succ 4$ & 2 \\
$2 \succ 1 \succ 5 \succ 4 \succ 3$ & 2 \\
$1 \succ 5 \succ 4 \succ 3 \succ 2$ & 2 \\
$5 \succ 4 \succ 2 \succ 3 \succ 1$ & 1 \\
$1 \succ 5 \succ 3 \succ 4 \succ 2$ & 1 \\
$2 \succ 1 \succ 4 \succ 5 \succ 3$ & 1 \\
$3 \succ 2 \succ 5 \succ 1 \succ 4$ & 1 \\
$4 \succ 3 \succ 1 \succ 2 \succ 5$ & 1 
\end{tabular}
&
\begin{tabular}{cc}
Agent type & Quantity \\
\hline
$5 \succ 4 \succ 3 \succ 2 \succ 1$ & 1 \\
$4 \succ 3 \succ 2 \succ 1 \succ 5$ & 1 \\
$3 \succ 2 \succ 1 \succ 5 \succ 4$ & 1 \\
$2 \succ 1 \succ 5 \succ 4 \succ 3$ & 1 \\
$1 \succ 5 \succ 4 \succ 3 \succ 2$ & 1 \\
$5 \succ 3 \succ 1 \succ 4 \succ 2$ & 2 \\
$1 \succ 4 \succ 2 \succ 5 \succ 3$ & 2 \\
$2 \succ 5 \succ 3 \succ 1 \succ 4$ & 2 \\
$3 \succ 1 \succ 4 \succ 2 \succ 5$ & 2 \\
$4 \succ 2 \succ 5 \succ 3 \succ 1$ & 2 
\end{tabular}
\end{tabular}
\end{center}
Running \cref{alg:greedy} on both Agent Cohort A and Agent Cohort B results in the same cycle: 
{\small \begin{multline*} \{1, 2\} \to \{1, 2, 3\}\to\{2, 3\}\to\{2, 3, 4\}\to\{3, 4\}\to\{3, 4, 5\}\to \\* \{4, 5\}\to \{1, 4, 5\}\to \{1, 5\}\to \{1, 2, 5\}\to \{1,2 \}.\end{multline*}}
However, the set of $t$-feasible and $u$-uncontestable $O \subseteq G$ for Agent Cohort A is 
\[ \mc{O}_A := \{ \{1,3\}, \{4,1\}, \{2,4 \}, \{5,2\}, \{3,5\} \}, \]
whereas the set of $t$-feasible and $u$-uncontestable $O \subseteq G$ for Agent Cohort B is 
\[ \mc{O}_B := \{ \{1,2,4\}, \{3,4,1\}, \{5,1,3\}, \{2,3,5\}, \{4,5,2 \}\}.\]
Since $\mc{O}_A \cap \mc{O}_B = \emptyset$, the cycle transcript alone cannot determine a $t$-feasible and $u$-uncon\-testable $O \subseteq G$: more intricate information about the precise quantities and types of agents is needed. 

Moreover, the following summary statistics also do not suffice to determine a $t$-feasible and $u$-uncontestable $O \subseteq G$, as they are identical for Agent Cohort A and Agent Cohort B: 
\begin{itemize}
\item Number of agents $n$ (we have $n=15$ for both)
\item Number of agents who rank public good $k$ in place $p$ (this quantity is $3$ for all $k$ and $p$, due to the cyclic symmetry of both agent bodies)
\end{itemize}
\end{example}

\section[Reformulation of the Covering Condition for Complete Preferences]{Reformulation of the Covering Condition for\texorpdfstring{\\}{ }Complete Preferences}\label{app:matrix-structure-complete-case}

In \cref{sec:computational-c56-proof}, we prove that every menu selection problem on $g \le 6$ public goods with $u \ge 2t\!-\!1$ has a $(t,u)$-stable menu via a reduction to a polyhedra covering problem. Specifically (\cref{thm:polyhedra-reduction}), we show that for fixed $g,t,u$, the existence of $(t,u)$-stable solutions for all menu selection problems over preferences $\mc{P}$ is equivalent to showing $\Z_{\geq 0}^{|\mc{P}|} \subseteq \bigcup_{O \subseteq G} P_O^{g,t,u}$, where $P_O^{g,t,u}$ are polyhedra depending only on $O, g,t,u, \mc{P}$. 

In this section we consider the special case of menu selection problems for which all agents have complete preferences over the public goods $G$. In this case, we show that the polyhedra $P_O^{g,t,u}$ have a simple structure. It may be possible to utilize this structure to efficiently determine whether the covering condition $\Z_{\geq 0}^{|\mc{P}|} \subseteq \bigcup_{O \subseteq G} P_O^{g,t,u}$ is satisfied for larger values of $g$. 

Fix $g,t,u$, and $G = \{ 1, \dots, g \}$. Let $\mc{P}$ be the set of all $g!$ linear orders of $G$. Fix the following ``lexicographic-cyclic'' ordering of $\mc{P}$: when $g=2$, the ordering is $[1\succ2, 2\succ1]$, when $g=3$, the ordering is $[1\succ2\succ3, 1\succ3\succ2, 2\succ3\succ1, 2\succ1\succ3, 3\succ1\succ2,3\succ 2\succ1]$, and so on. A general definition is below. 

\begin{definition}[``Lexicographic-cyclic'' ordering]
Let $g \in \N$ and $G = \{1, \dots, g \}$. Let $\mc{P}$ denote the set of all $g!$ linear orders of $G$. Let $\oplus_g$ denote addition modulo $g$ (where we take the result to be in $\{1,2,\dots, g \}$). Then the \textit{lexicographic-cyclic ordering} of $\mc{P}$, which we denote $(\succ_G^{(i)})_{i \in \{1, \dots, g!\}}$, is defined recursively as follows. 
\begin{itemize}
    \item If $g=2$, then $\succ_G^{(1)} := 1\succ2$, and $\succ_G^{(2)} := 2\succ1$. 
    \item Now suppose $g \ge 3$. Pick $i \in \{1, \dots, g!\}$. Write $i = (g\!-\!1)!(i_1-1) + i_2$ for $i_1, i_2 \in \{1, \dots, c\}$. Then set $\succ_G^{(i)} := i_1 \succ [\succ_{-j \oplus_g O \setminus j}^{i_2} \oplus_g j]$.\footnote{This is well-defined because $-j \oplus_g O \setminus j = \{1, \dots, g\!-\!1\}$.}
\end{itemize}
\end{definition}

Recall (see \cref{sec:computational-c56-proof}) that each polyhedron $P_O^{g,t,u}$ is given by $P_O^{g,t,u} = \{ x \in \Z^{|\mc{P}|}_{\geq 0}:\ D A_O^{g,t,u} x \geq b \}$, where $D$ is a simple $g \times g$ diagonal matrix ($D^{(ii)} = 1$ if $i \in O$ and $-1$ otherwise) and $b$ is a simple $g \times 1$ vector ($b^{(i)} = t$ if $i \in O$ and $-(u\!-\!1)$ otherwise). Thus, most of the complexity in $P_O^{g,t,u}$ comes from $A_O^{g,t,u}$. The following result demonstrates how to specify the matrix $A_O^{g,t,u}$ as a recurrence relation involving simple matrix operations.

\begin{lemma}\label{lemma-calculation}
First we fix some notation\footnote{Note: in this lemma only, we use $i$ to denote an index, rather than an agent.}:
\begin{itemize}
\item Let $e_i \in \R^g$ be a row vector with $i$th entry $1$ and all other entries $0$.
\item Let $\sigma$ be the permutation $(12\cdots g)$ on $G$. Let $P_\sigma$ be the $g\times g$ permutation matrix of $\sigma$.
\item For all $i \in G$, let $f_i: [g\!-\!1] \to G\setminus \{i \}$ be the unique bijection such that $f_i(k) \equiv k + i \mod g$. (So, for example, for $g=4$ and $i=2$, we have $f_2(1) =3, f_2(2)=4, f_2(3)=1$.)
\end{itemize}
  Then, for all $g \geq 2$, and for all $O \subseteq G$,
\[ A_O^{g,t,u} = I_g \otimes \underbrace{\begin{pmatrix}1 & \cdots & 1\end{pmatrix}}_{g!}  + \sum_{i \in G \setminus O} e_i \otimes \left[ P_{\sigma^{i-1}} \begin{pmatrix}\mathbf{0} \\ I_{g\!-\!1} \end{pmatrix} A^{g\!-\!1, t, u}_{f_i^{-1}(O)} \right] \]
\end{lemma}
\begin{proof}
Consider the $ij$-entry of $A_O^{g,t,u}$, denote it $a_{ij}$. Write $j := (g\!-\!1)! (j_1-1) + j_2$, for $j_2 \in \{1, \dots, (g\!-\!1)! \}$ and $j_1 \in \{1, \dots, g \}$. 

Case 1: $j_1  \in O$. Then the preference relation $\succ$ corresponding to $j$ ranks $j$ first. This means $a_{ij} = 1$ if $i = j_1$ and $a_{ij} = 0$ otherwise. The left term on the right-hand side has $ij$-entry equal to 1 if $i = j_1$ and 0 otherwise. The right term on the right-hand side has $ij$-entry equal to 0, because are summing over terms $e_{j_1} \otimes \left[ \text{something} \right]$ for $j_1 \in G \setminus O$, so only columns with $j_1 \in G \setminus O$ are affected. 

Case 2: $j_1 \notin O$ and $i = j_1$. Then the preference relation $\succ$ corresponding to $j$ ranks $i$ first, but $i \notin O$. Thus $a_{ij} = 1$, because $i$ is preferred over $O$ to under $\succ$. The left term on the right-hand side has $ij$-entry equal to 1, because $i = j_1$. The right term on the right-hand side has $ij$-entry equal to 0, because $P_{\sigma^{i-1}} \begin{pmatrix}\mathbf{0} & I_{g\!-\!1} \end{pmatrix}^T$ has an $i$th row consisting of all zeros. 

Case 3: $j_1 \notin O$ and $i \neq j_1$. Let $B_{j_1} \in \Z^{(g\!-\!1) \times (g\!-\!1)!}$ be the result of removing the $i$th row and all columns of the form $(g\!-\!1)! (j_1' - 1) + j_2$ from $A_O^{g,t,u}$, for $j_1' \neq j_1$, $j_1' \in \{1, \dots, g \}$, and $j_2 \in \{1, \dots, (g\!-\!1)! \}$. Notice on the right-hand side, the left term does not contribute to any of the indices in $B_{j_1}$. Similarly, on the right-hand side, the only summand in the right term which contributes to any of the indices in $B_{j_1}$ is $e_{j_1} \otimes \left[ P_{\sigma^{j_1-1}} \begin{pmatrix}\mathbf{0} & I_{g\!-\!1} \end{pmatrix}^T A^{g\!-\!1, t, u}_{f_{j_1}^{-1}(O)} \right]$. 

So it suffices to show $B_{j_1}$ is equal to $P_{\sigma^{j_1-1}} \begin{pmatrix}\mathbf{0} & I_{g\!-\!1} \end{pmatrix}^T A^{g\!-\!1, t, u}_{f_{j_1}^{-1}(O)}$ with the $i$th row removed (which, recall from Case 2, is just a zero row). This is a matter of relabeling. In $B_{j_1}$, the columns refer to all preferences with top choice $j_1$, and the rows refer to all public goods except $j_1$. So we are left with a copy of $A^{g\!-\!1, t, u}$ on the menu $G \setminus \{ j_1 \}$, instead of $\{1, \dots, g\!-\!1 \}$, and so when we reorder the rows (with $I_{g\!-\!1}$) and relabel the public goods ($f_{j_1}^{-1}(O)$), we obtain a matrix equality.
\end{proof}

\section{Validity of Improved Lower Bound for Complete Preferences}\label{app:c7example-complete}

In \cref{lower-bound-2311}, we show that for all $g \ge 7$ and $u \leq 23 \lfloor \frac{t-1}{11} \rfloor$, there exists a menu selection problem for which no $(t,u)$-stable solution exists. We do so by constructing a set of $70x$ agents, for which there exists no $(11x+1, 23x)$-stable menu of public goods, where $x := \lfloor \frac{t-1}{11} \rfloor$ (see \cref{tab:better-lower-bound}). 

In this construction, each agent only ranks their top three or four choices. One might ask whether this result holds for the case of complete preferences. (In particular, notice that $t$-feasibility is more difficult to violate in the case of complete preferences than for incomplete preferences.) In fact, all possible ways of ``filling in'' the preferences in \cref{tab:better-lower-bound} by adding missing public goods to the bottom of the preference list yield a menu selection problem with complete preferences for which there exists no $(11x+1, 23x)$-stable solution. 

Let $\succ$ refer to the original truncated preferences, and $\succ'$ refer to the (arbitrarily) completed preferences. Set $u' := 23x$ and $t' := 11x+1$. By \cref{lower-bound-2311}, $\problem$ does not have a $(t', u')$-stable solution. We aim to show that $\extproblem$ does not have a $(t', u')$-stable solution. 

Similarly to \cref{sec:improved-lower-bound}, by monotonicity (\cref{lemma:monotonicity}), it suffices to show the following.
\begin{enumerate}[(1)]
    \item Every $O \subseteq \{1, 2,3,4,5,6,7\}$ with $|O|=3$ is not $u'$-uncontestable for $\extproblem$. 
    \item Every $O \subseteq \{1, 2,3,4,5,6,7\}$ with $|O|=4$, except $O$ of the form $\{k,k\!+\!1,k\!+\!2,k\!+\!3\}$\footnote{With wraparound modulo 7 when applicable.}, is not $t'$-feasible for $\extproblem$. 
    \item Every $O \subseteq \{1,2,3,4,5,6,7\}$ with $|O|=4$ of the form $\{k,k\!+\!1,k\!+\!2,k\!+\!3\}$ is not $u'$-uncontestable for $\extproblem$. 
    \item Every $O \subseteq \{1,2,3,4,5,6,7\}$ with $|O|=5$ is not $t'$-feasible. 
\end{enumerate}

First we show (1). Notice for all $O \subseteq \{1, 2,3,4,5,6,7\}$ with $|O|=3$ and all $k \notin O$ we have $|k \succ O| \leq |k \succ' O|$. Since all such $O$ are not $u'$-uncontestable for $\problem$, they are also not $u'$-uncontestable for $\extproblem$. 

Next we show (2). Consider any $O$ of the form specified in (2). Then every agent has an element of $O$ ranked in their original preferences (before filling in). Thus for all $k \in O$, $|k \succ O\setminus \{ k \}| = |k \succ' O \setminus \{ k \} |$. Since all such $O$ are not $t'$-feasible for $\problem$, they are also not $t'$-feasible for $\extproblem$. 

Next we show (3). Consider any $O$ of the form specified in (3). We claim $|k-1 \succ' O| \geq u'$. Notice every agent has an element of $O \cup \{k-1\}$ ranked in their original preferences (before filling in). Thus $|k-1 \succ' O| = |k-1 \succ O| = |3 \succ \{4,5,6,7\}| = 5x+5x+5x+3x+3x+3x+x+x=26x \geq u'$.

Finally we show (4). Then as in (2), for any $O$ with $|O|=5$, every agent has an element of $O$ ranked in their original preferences (before filling in). Thus for all $k \in O$, $|k \succ O\setminus \{ k \}| = |k \succ' O \setminus \{ k \} |$. Since all such $O$ are not $t'$-feasible for $\problem$ (by \cref{lemma:monotonicity} and \cref{lemma:o4}), they are also not $t'$-feasible for $\extproblem$.

\section[Upper Bound with Linear-Time Algorithm via Structural Results]{Upper Bound with Linear-Time Algorithm via\texorpdfstring{\\}{ }Structural Results}\label{app:structural-upper}

In this section, we explore the possibility of obtaining upper bounds via structural results. Specifically, we show that when the ratio between $u$ and $t$ is sufficiently large, either there exists a $(t,u)$-stable menu of size $\leq 1$, or the agents' preferences must have a highly specific structure. From this specific structure, we can compute a $(t,u)$-stable solution using a simple algorithm. The ratio between $u$ and $t$ needed to secure such a structural result exceeds that of \cref{upper-bound-g-minus-2-formal}, and as a result, this technique does not yield an improved upper bound for the case of complete preferences. For the case of general, possibly incomplete preferences, the result we present in this section, \cref{prop:structural-upper-bound}, is the best upper bound we have for general $g$, as we do not know how to extend the techniques of \cref{upper-bound-g-minus-2-formal} to the setting of incomplete preferences. Furthermore, the tools employed here may shed additional light on the menu selection problem. 

As a warmup, consider the case $u \geq (g\!-\!1)(t\!-\!1)$. (This is slightly stronger than the trivial upper bound presented in \cref{upper-bound-g-minus-1}.) Assume as usual (by applying \cref{lemma:suff-popular} and \cref{lemma:suff-rarely-ranked}) that $|k| < t$ and $\{ k \}$ is $t$-feasible for all $k \in G$. Suppose no $(t,u)$-stable menu of size $\leq 1$ exists. Then for every public good $k \in G$, there exists a public good $j \in G$ such that $|j \succ k| \geq u \geq (g\!-\!1)(t\!-\!1)$. The agents with favorite public good $k$ cannot satisfy $j \succ k$, so every agent whose favorite public good is not $k$ must satisfy $j \succ k$. Repeating this argument for all $k \in G$, and relabeling the public goods, we get that every agent with favorite public good $k$ must satisfy all of $2 \succ 1, 3 \succ 2, \dots, 1 \succ g$ except $k+1 \succ k$. This imposes a rigid structure on the agents' preferences: it implies that the set of agents has preferences exactly as follows (see \cref{table:structure-warmup}). 

\begin{table}[h!t]\small 
\centering 
\begin{tabular}{c|c}
        Quantity & Preference \\
        \hline 
        $t\!-\!1$ & $1 \succ \ \ g\ \ \ \succ g\!-\!1 \succ \dots \succ 3 \succ 2$ \\
        $t\!-\!1$ & $2 \succ \ \  1\ \  \  \succ \ \ g \ \ \ \succ \dots \succ 4 \succ 3$ \\
        $t\!-\!1$ & $3 \succ \ \ 2\ \ \   \succ \ \ 1 \ \ \ \succ \dots \succ 5 \succ 4$ \\
        $\vdots$ & \hspace{-0.5cm} $\vdots$ \\
        $t\!-\!1$ & $g \succ g\!-\!1 \succ g\!-\!2 \succ \dots \succ 2 \succ 1$
    \end{tabular}
\caption{When $u \geq (g\!-\!1)(t\!-\!1)$, and no $(t,u)$-stable menu of size $\leq 1$ exists, then the agents' preferences must be exactly of this form.} 
\label{table:structure-warmup}
\end{table}

The following lemma shows that we can still guarantee much structure in the agents' preferences if the ratio between $u$ and $t$ is relaxed to $u \geq (g\!-\!1-\epsilon)(t\!-\!1)$. 

\begin{definition}
Let $\problem$ be a menu selection problem and $G' \subseteq G$ with $G' = \{j_1, \dots, j_m \}$ for some $2 \leq m \leq g$. An agent $i \in N$ is said to be \textit{regular} (with respect to $G'$) if they have favorite public good $j_k \in G'$ and satisfy all of $j_2 \succ_i j_1, j_3 \succ_i j_2, \dots, j_1 \succ_i j_m$ except $j_{k+1} \succ_i j_k$. Otherwise, an agent is called \textit{irregular} (with respect to $G'$).
\end{definition}

\begin{lemma}[Structural lemma]\label{lemma:structural-lemma}
Let $g \geq 5$, $t \geq 2$, and $u \geq (g\!-\!1-\epsilon)(t\!-\!1)$ for any $\epsilon \geq 0$. Then at least one of the following holds:
\begin{enumerate}[(1)]
    \item There exists some $O \subseteq G$ with $|O| \leq 1$ such that $O$ is $(t,u)$-stable. 
    \item There exists a subset of public goods $G' \subseteq G$, with $|G'| = m$ for some $2 \leq m \leq g$, such that $\leq m\epsilon (t\!-\!1)$ agents are irregular with respect to $G'$. 
\end{enumerate}
\end{lemma}

\begin{proof}
See \cref{sec:proofs-structural-lemmas}.
\end{proof}

\cref{lemma:structural-lemma} implies that in either case, we can find a $(t,u)$-stable solution using a simple algorithm. In the case of (1), we can find a $(t,u)$-stable $O$ by brute force (only $g\!+\!1$ menus to check). Otherwise, if the agents' preferences have the structure of (2), then we can find a $(t,u)$-stable menu using \cref{alg:greedy-structured}.

\begin{algorithm}
\caption{Computing a $(t,u)$-stable solution from structured preferences}\label{alg:greedy-structured}
\begin{algorithmic}
\State Given: agents' preferences with the structure of (2) in \cref{lemma:structural-lemma}.
\State Without loss of generality, relabel public goods so that $G' = \{1, \dots, m \}$.
\State $x_k \gets$ number of regular agents with favorite public good $k$. 
\State Initialize $O\gets \emptyset$.
\State Set $p \gets 1$. (This is the currently `pending' public good.) 
\For{$k = 2, 3, \dots, m$}
\If{$x_p + \dots + x_k \geq t$}
\State Set $O \gets O \cup \{ p \}$. 
\State Set $p \gets k + 1$.
\EndIf
\EndFor\\
\Return $O$.
\end{algorithmic}
\end{algorithm}

\begin{lemma}\label{lemma:stable-set-from-structure}
Let $g \geq 5$, $t \geq 2$, and $u \geq (g\!-\!1-\epsilon)(t\!-\!1)$ with $\epsilon < \frac16$. Then \cref{alg:greedy-structured} outputs a $(t,u)$-stable menu $O$ with $|O| \geq 2$. 
\end{lemma}

\begin{proof}
See \cref{sec:proofs-structural-lemmas}.
\end{proof}

\cref{lemma:structural-lemma} and \cref{lemma:stable-set-from-structure} immediately imply the following upper bound.

\begin{proposition}\label{prop:structural-upper-bound}
Let $g \geq 5$, $t \geq 2$ and $u \geq (g\!-\!1-\epsilon)(t\!-\!1)$, where $\epsilon < \frac16$. Then every menu selection problem with $g$ public goods has a $(t,u)$-stable solution. 
\end{proposition}

\subsection{Proofs of \cref{lemma:structural-lemma} and \cref{lemma:stable-set-from-structure}}\label{sec:proofs-structural-lemmas}

\begin{proof}[Proof of \cref{lemma:structural-lemma}]
By \cref{lemma:suff-popular}, we can assume $|k| < t$ for all $k \in G$. Suppose there does not exist a $(t,u)$-stable menu $O \subseteq G$ with $|O| \leq 1$. Then for all $k \in G$, there exists some $j \in G$ such that $|j \succ k| \geq u$. Apply this operation successively to obtain, for some $m \in \{2, 3, \dots, g \}$, distinct public goods $j_1, \dots, j_m \in G$ such that $|j_{i+1} \succ j_i| \geq u$ for each $i \in \{1,\dots, m \}$ (set here $m+1 := 1$). Let $G' := \{j_1, \dots, j_m \}$. 

Now we bound the number of irregular agents from above. For each $j_i \in G'$, the total number of agents for whom $j_i$ is not their favorite public good is $\sum_{\ell \neq j_i} |\ell|$. Also, there are $\geq u$ agents with $j_{i+1} \succ j_i$. Thus
\begin{align*}
|j_i \succ j_{i+1}| - |j_i| &\leq  \sum_{\ell \neq j_i} |\ell| - u \\
&\leq (g\!-\!1)(t\!-\!1) - (g\!-\!1\!-\!\epsilon)(t\!-\!1) \\
&= \epsilon(t\!-\!1). 
\end{align*}
Summing over all $i \in \{1, \dots, m \}$ implies that there are at most $m\epsilon(t\!-\!1)$ irregular agents.
\end{proof}

\begin{proof}[Proof of \cref{lemma:stable-set-from-structure}]
Let $O$ be the menu outputted by \cref{alg:greedy-structured}. We would like to show that $O$ is $(t,u)$-stable. 

Recall that $x_k$ denotes the number of regular agents with favorite $k$, that is, agents who satisfy all of $2 \succ 1, 3 \succ 2, \dots, 1 \succ m$ except $k+1 \succ k$. Let $B$ denote the number of irregular agents. By \cref{lemma:structural-lemma}, $B \leq m\epsilon(t\!-\!1)$. 

First we show $|O| \geq 2$. Suppose $|O| \leq 1$. If $O = \emptyset$, then $x_1 + \dots + x_m < t$. If $O = \{ k \}$ for some $k \in G$, then $x_1 + \dots + x_{k-1} < t$ and $x_1 + \dots + x_{k} \geq t$. Additionally, because no other public good was added to $O$ after $k$, we have $x_{k+1} + \dots + x_m < t$. By \cref{lemma:suff-popular}, we can assume $x_k < t$. Thus the number of regular agents is at most $ 3(t\!-\!1)$. This implies that the total number of agents is $n \leq 3(t\!-\!1) + B\leq (m\epsilon + 3)(t\!-\!1)$. Additionally, $n \geq (g\!-\!1-\epsilon)(t\!-\!1)$. But $3+m\epsilon < g\!-\!1-\epsilon$ because $\epsilon < \frac{g-4}{g\!+\!1}$ whenever $\epsilon < \frac16$ and $g \geq 5$. Thus we have a contradiction, so $|O| \geq 2$. 

Now we show that $O$ is $(t,u)$-stable.  

Showing $t$-feasibility: Pick some $k \in O$. Let $k'$ be the least element of $O$ strictly greater than $k$, or $m+1$ if no such element exists. Then by the algorithm construction, $x_k + x_{k+1} + \dots + x_{k' - 1} \geq t$. Because a regular agent with favorite $\ell$ has type $\ell \succ \ell - 1\succ \dots \succ 1 \succ m \succ \dots \succ \ell + 1$ (possibly with public goods in $G \setminus G'$ inserted in this ranking), the regular agents with favorite $k, k+1, \dots, k' - 1$ prefer $k$ over any other public good in $O$. Without even using the irregular agents, we have shown that public good $k$ has at least $t$ agents. 

Showing $u$-uncontestability: Pick some $k \notin O$. If $k \in G \setminus G'$, then since $|O| \geq 2$, the maximum lobby size is $|k \succ O| \leq \sum_{\ell \in G \setminus O} |\ell| \leq (g\!-\!2)(t\!-\!1) < u$. It remains to consider the case $k \in G' \setminus O$. Again let $k'$ be the least element of $O$ strictly greater than $k$, or $g\!+\!1$ if no such element exists. Then the number of regular agents lobbying for $k$ over $O$ is $x_{k} + \dots + x_{k' - 1}$. Let $k''$ be the greatest element of $O$ strictly less than $k$. (Such $k''$ exists because $1 \in O$, since $O \neq \emptyset$.)

Case 1: $k' \in O$. Because $k'$ is the smallest value for which $x_{k''} + \dots + x_{k'-1} \geq t$, we have $x_{k''} + \dots + x_{k'-2} < t$. Also, $x_{k'-1} < t$. Thus 
\[ x_{k} + \dots + x_{k' - 1} \leq (x_{k''} + \dots + x_{k'-2}) + x_{k' - 1} \leq 2(t\!-\!1). \]

Case 2: $k' = m+1$. Then $k$ is larger than any element of $O$. Let $p$ be the element that was pending when $k''$ was added. Since $p$ was pending but never added, $x_p + \dots + x_m < t$. Since $k$ was added, $x_k + \dots + x_{p-1} \geq t$ but $x_k + \dots + x_{p-2} < t$. Also, $x_{p-1} < t$. Thus 
\[ x_{k} + \dots + x_{k' - 1} \leq (x_k + \dots + x_{p-2}) + x_{p-1} + (x_p + \dots + x_m) \leq 3(t\!-\!1). \]

Thus, in either case, the number of regular agents lobbying for $k$ over $O$ is $x_k + \dots + x_{k' - 1} \leq 3(t\!-\!1)$. Thus, the number of agents lobbying for $k$ over $O$ is 
\begin{align*}
|k \succ O| &\leq 3(t\!-\!1) + B \\
&\leq 3(t\!-\!1) + m\epsilon(t\!-\!1) \\
&= (3 + m\epsilon)(t\!-\!1) \\
&< (g\!-\!1-\epsilon)(t\!-\!1) \\
&\leq u, 
\end{align*}
where here we observe $3 + m\epsilon < g\!-\!1-\epsilon$ because $\epsilon < \frac{g-4}{g\!+\!1} = 1 - \frac{5}{c+1}$ whenever $\epsilon < \frac16$ and $g \geq 5$.
\end{proof}

\section{Connection to Core Stability and NTU Games}\label{app:core-ntu}

In this paper, we pose a new kind of matching problem and give conditions on problem parameters $g,t,u \in \mathbb{N}$ such that all menu selection problems on $g$ goods admit a $(t,u)$-stable menu. Our notion of stability, and specifically the condition of $u$-uncontestability---that no $u$ agents prefer some unprovided good $j \in G \setminus O$ over the provided menu $O \subseteq G$---may remind certain readers of the notion of \textit{core stability} from cooperative game theory. Indeed, our matching problem can be formulated as an NTU game, the core of which precisely corresponds to what we call $(t,u)$-stable menus. This raises the question as to whether existing techniques used to study the core of NTU games might also yield results in our setting. Accordingly, in this section, we first formally show how to formulate our matching problem as an NTU game, and second show how \textit{balancedness}, a core technique (no pun intended) for proving (non)emptiness of the core, falls short in our setting.

\newcommand{\ntugame}{\ensuremath{\bigl(\mc{N}, \mc{X}, \mc{V}, (\succsim_i')_{i \in \mc{N}}\bigr)}}

\subsection{Formulation as an NTU Game}\label{subsec:formulation-as-ntu}

We follow the notation of \citet[][section 13.5]{osborne_course_1994}. An \emph{NTU game} is specified by: a set of agents $\mc{N}$, a set of consequences $\mc{X}$, a value function $\mc{V}: 2^\mc{N} \to 2^\mc{X}$ that maps subsets of agents to subsets of consequences, and for each player $i \in \mc{N}$, a (weak) preference relation $\succsim_i'$ over the consequences $\mc{X}$. 

We begin by describing a procedure for converting an arbitrary menu selection problem into a corresponding NTU game. For simplicity, throughout this section, we assume that $n > u$ (recall $n$ denotes the number of agents).

\begin{definition}[A menu selection problem as an NTU game]\label{defn:convert-ntu}
Let $\problem$ be a menu selection problem and let $t,u \in \N$ such that $n > u$. The \textit{corresponding NTU game}, denoted \ntugame, is defined as follows:
\begin{itemize}
    \item Set $\mc{N} := N$ and $\mc{X} := 2^G$ (each agent is an agent, and each consequence is a menu).
    \item For each agent $i \in \mc{N}$,
    and for every two menus $O,O' \in \mc{X}$, define that $O \succ_i O'$ if and only if agent $i$'s favorite good in $O$ is strictly better than agent $i$'s favorite good in $O'$.
    \item The value function $\mc{V} : 2^N \to 2^{\mc{X}}$ is defined as follows:
    \[
\mc{V}(S) := 
\begin{cases}
    \emptyset & \text{ if }|S| < u,\\
    \text{all single-item menus that are $u$-feasible w.r.t.\ $S$} & \text{ if }u \le |S| \le n - 1,\\
    \text{all $t$-feasible menus} & \text{ if }S = \mc{N}.
\end{cases}
\]
\end{itemize}
\end{definition}

We next recall the definition of the core of an NTU game, and then prove that the set of $(t,u)$-stable menus for a menu selection problem coincides with the core of the corresponding NTU game. 

\begin{definition}[Core of an NTU game; see \cite{osborne_course_1994}]
    The \emph{core} of the NTU game \ntugame is the set of all $O \in \mc{V}(N)$ for which there is no coalition $S \in \mc{X}$ and $O' \in \mc{V}(S)$ such that $O' \succ_i' O$ for all $i \in S$. 
\end{definition}

\begin{proposition}\label{prop:ntu-equivalence}
Let $\problem$ be a menu selection and let $t,u \in \mathbb{N}$ with $u > t \ge 2$. Let  \ntugame be the corresponding NTU game.
Let $\mathcal{O} \subseteq 2^G$ be the set of $(t,u)$-stable menus of $\problem$ and let $\mc{O}' \subseteq V(\mc{N})$ be the core of the corresponding NTU game. Then $\mc{O} = \mc{O}'$.
\end{proposition}

\begin{proof}
Let $O$ be a menu. To show that $O$ is $(t,u)$-stable if and only if $O$ is in the core, it suffices to show that both of the following hold.
\begin{enumerate}
\item $O$ is $t$-feasible if and only if $O\in\mc
{V}(\mc{N})$,
\item $O$ is $u$-uncontestable if and only if there does not exist a coalition $S\in\mc{X}$ and $O'\in\mc{V}(S)$ such that $O' \succ_i' O$ for all $i \in S$.
\end{enumerate}
Item 1 follows by definition of $\mc
{V}(\mc{N})$. For Item 2, note that:
\begin{align*}
&\text{$O$ is not $u$-uncontestable}
\\ & \qquad \Longleftrightarrow \\
&\text{there exists a coalition $S\subseteq N$ with $|S|\ge u$ and a good $k\in G$ such that $k\succ_i O$ for all $i\in S$} \\
&\qquad \Longleftrightarrow \text{(set $O':=\{k\}$)} \\
&\text{there exists a coalition $S\in\mc{X}$ and $O'\in\mc{V}(S)$ such that $O' \succ_i' O$ for all $i \in S$.}\qedhere
\end{align*}
\end{proof}

Given $t,u\in\N$, due to \cref{prop:ntu-equivalence}, when no confusion can arise we sometimes conflate a menu selection problem with its corresponding NTU game. Specifically, we might refer to the the core of the NTU game as the core of the menu selection problem or say that the menu selection problem is balanced to mean that the NTU game is balanced (to be defined below).

\subsection{Inadequacy of Balancedness for Characterizing Stability}

In \cref{sec:simple-lower-upper-bounds,sec:small-c,sec:beyondc6}, we characterize the parameters $t,u,g \in \N$ for which it is the case that all menu selection problems over $g$ goods have a $(t,u)$-stable menu. Rephrased in terms of the core (by \cref{prop:ntu-equivalence}), we characterize for which parameters $t,u,g \in \N$ it is the case that for all menu selection problems over $g$ goods have a nonempty core (with respect to~$t,u$). 

There is a rich prior literature in cooperative game theory that seeks conditions that imply (non)emptiness of the core of NTU games. For example, in the case of TU games (a specific subclass of NTU games), nonemptiness of the core is equivalent to \emph{balancedness} \citep{bondareva_applications_1963,shapley_1967_on}. For NTU games, characterizing nonemptiness of the core is more complex \citep[see, e.g.][for a overview]{peleg_introduction_2007}.

As a case study, we consider a canonical sufficient condition for nonemptiness of the core of an NTU game proposed by \citet{scarf_core_1967}. Specifically, \citet{scarf_core_1967} shows that the core of a NTU game is nonempty if it is \textit{balanced} (a generalization of balancedness for TU games). Unlike in the setting of TU games (where balancedness is equivalent to nonemptiness of the core), in the setting of NTU games, balancedness is only a sufficient condition for the nonemptiness of the core of a NTU game.\footnote{See \citet{peters_2003_ntu} for a simple example of an NTU game in which balancedness is not necessary for nonemptiness of the core.} To demonstrate that balancedness gives a coarser sufficient condition for the nonemptiness of the core of a menu selection problem in our setting than our characterization does, we exhibit a menu selection problem that is not balanced yet has a nonempty core (i.e., has a stable menu). 

To apply the result of \citet{scarf_core_1967}, we first follow that paper by representing consequences, without loss of generality, as utility vectors in $\R^n$, where the $i$th component corresponds to the utility of agent $i$ when menu~$O$ is offered. We furthermore adopt the following convention, also without loss of generality: If an agent is assigned to their first-choice good, they have utility $g$, if to their second choice, they have utility $g-1$, and so on (the proof would go through in precisely the same way for any other monotone choice of utilities for each agent, and even with different choices for different agents). We denote the resulting NTU game by $\bigl(\mc{N}, \mcb{X}, \mcb{V}, (\succsim_i'')_{i \in \mc{N}}\bigr)$.\footnote{Note that $\succ_i''$ has a simple description: Given two consequences $x,x' \in \mcb{X}$ (utility vectors), it compares the $i$th elements of these vectors.}

We now recall \citeauthor{scarf_core_1967}'s (\citeyear{scarf_core_1967}) notion of balancedness for NTU games. Here we use the following notation: $\text{Proj}_S(x)$ denotes projection of the vector $x$ onto its components in $S$, and $\Down(A)$ denotes the downward closure of the set $A \subseteq \R^n$.\footnote{That is, $\Down(A) := \bigl\{ x  \in \R^n~\big|~\ \exists a \in A \text{ s.t. }x^{(i)} \le a^{(i)}\ \forall i \in [n] \bigr\}$.}

\begin{definition}[Balancedness for NTU games; see \cite{scarf_core_1967}]
Let $\mc{S} \subseteq 2^\mc{N}$ be a collection of coalitions. Weights $\{\lambda_S\}_{S \in \mc{S}} \in [0,1]$ of $\mc{S}$ are said to be \textit{balanced} if, for every agent $i \in \mc{N}$, we have $\sum_{S \in \mc{S}, S \ni i} \lambda_S = 1$. Accordingly, the collection $\mc{S}$ is said to be \textit{balanced} if there exist balanced weights for $\mc{S}$. Finally, an NTU game is said to be \textit{balanced} if for every balanced collection $\mc{S}$, for all $x \in \mcb{X}$, if $\text{Proj}_S(x) \in \Down\bigl(\mcb{V}(S)\bigr)$ for all $S \in \mc{S}$, then $x \in \Down(\mcb{V}\bigl(\mc{N})\bigr)$.
\end{definition}

\begin{theorem}[\cite{scarf_core_1967}]
A balanced NTU game has a nonempty core.
\end{theorem}

Finally, we are ready to exhibit an example of a menu selection problem that is not balanced yet has a nonempty core. Set $t,u \in \N$ with $t \ge 2$ and $u = 2t-1$. We consider $n := 6t-3$ agents with the preferences as listed in \cref{tab:ntu-example}.

\begin{table}
\centering
\begin{tabular}{c|c|c}

$\#$ & Preference & Coalition Membership \\
\hline 
$t-1$ & $1 \succ 2 \succ 3$ & $S_1$ \\
$t$ & $2 \succ 1 \succ 3$ & $S_1$ \\
\hline 
$2t-1$ & $2 \succ 3 \succ 1$ & $S_2$ \\
\hline 
$t-1$ & $3 \succ 2 \succ 1$ & $S_3 $\\
$t$ & $2 \succ 3 \succ 1$ & $S_3$\\
\end{tabular}
\caption{We consider $t-1$ agents with preferences $1 \succ 2 \succ 3$, $t$ agents with preferences $2 \succ 1 \succ 3$, and so on. For showing that this instance is not balanced, we consider a particular partition $\mc{N} = S_1 \sqcup S_2 \sqcup S_3$ of the agents, indicated here.}
\label{tab:ntu-example}
\end{table}

First, observe the core is nonempty, because $\{ 2 \}$ is $(t,u)$-stable. It remains to show that this menu selection problem is not balanced. Consider the balanced collection $\{ S_1, S_2, S_3 \}$, where the sets $S_1,S_2, S_3 \subseteq \mc{N}$ are as specified in \cref{tab:ntu-example}. (This collection is balanced because the weights $\lambda_{S_1} = \lambda_{S_2} = \lambda_{S_3} = 1$ are balanced.) 

We now compute $\mcb{V}(S_1), \mcb{V}(S_2), \mcb{V}(S_3)$, and $\mcb{V}(\mc{N})$. For notational ease, as agents with the same preferences have the same utilities from the same consequences, we write utility vectors in $\R^n$ as elements of $\R^5$. 
\begin{itemize}
    \item The only $u$-feasible single-item menus for $S_1$ are $\{1\}$ or $\{2\}$. These induce utilities $\begin{pmatrix}3&2&*&*&*\end{pmatrix}$ and $\begin{pmatrix}2&3&*&*&*\end{pmatrix}$ respectively. 
    \item The only $u$-feasible menu for $S_2$ is $\{ 2 \}$. This induces utilities $\begin{pmatrix}*&*&3&*&*&\end{pmatrix}$.
    \item The only $u$-feasible single-item menus for $S_3$ are $\{2\}$ or $\{3\}$. These induce utilities $\begin{pmatrix}*&*&*&2&3\end{pmatrix}$ and $\begin{pmatrix}*&*&*&3&2\end{pmatrix}$ respectively.
    \item The maximal $t$-feasible menus for $\mc{N}$ are $\{2 \}$ and $\{1,3 \}$. These induce utilities $\begin{pmatrix}2&3&3&2&3\end{pmatrix}$ and $\begin{pmatrix}3&2&2&3&2\end{pmatrix}$ respectively. (Non-maximal $t$-feasible menus do not affect the downward closure $\Down\bigl(\mcb{V}(\mc{N})\bigr)$.)
\end{itemize}

Now consider the utility vector $x = \begin{pmatrix}3&2&3&3&2\end{pmatrix}$. Then $\text{Proj}_{S_1}(x) \in S_1$ (same utilities as $\{1 \}$), $\text{Proj}_{S_2}(x) \in S_2$ (same utilities as $\{2 \}$), and $\text{Proj}_{S_3}(x) \in {S_3}$ (same utilities as $\{3 \}$). However, $x \notin \Down\bigl(\mcb{V}(\mc{N})\bigr)$. Thus, this menu selection problem is not balanced.

\end{document}